\documentclass[3p,final,number]{elsarticle}
\usepackage{mathpazo,charter,mathrsfs,amssymb,paralist}
\usepackage[raiselinks=false,colorlinks=true,citecolor=blue,urlcolor=blue,linkcolor=blue,bookmarksopen=true]{hyperref}
\usepackage[usenames,dvipsnames]{xcolor}
\usepackage{bookmark}
\usepackage[normalem]{ulem}
\usepackage[all,dvips,color,arrow,curve,poly]{xy}
\UseCrayolaColors
\newtheorem{theorem}{Theorem}
\newtheorem{proposition}{Proposition}
\newtheorem{definition}{Definition}
\newtheorem{corollary}[theorem]{Corollary}
\newtheorem{lemma}[theorem]{Lemma} 
\newenvironment{proof}{\noindent{\bf Proof.}\ }{\hfill\qed\par\bigskip}
\makeatletter\newcommand\setreflabel[1]{\protected@edef\@currentlabel{#1}}\newcounter{claim}

\def\ps@pprintTitle{\let\@oddhead\@empty\let\@evenhead\@empty\def\@oddfoot{\footnotesize\itshape\hfill\today}\let\@evenfoot\@oddfoot}
\makeatother\hyphenpenalty=5000\tolerance=1000

\renewcommand{\L}[0]{\ensuremath{\mathsf{L}}}

\newcommand{\bC}[0]{\mathbb{C}}
\newcommand{\bK}[0]{\mathbb{K}}

\newcommand{\bH}[0]{\mathbb{H}}
\newcommand{\bP}[0]{\mathbb{P}}

\renewcommand{\P}{\ensuremath{{\mathbb{P}}}}
\newcommand{\D}{\ensuremath{{\cal D}}}
\newcommand{\G}{\ensuremath{{\mathscr G}}}
\newcommand{\I}{\ensuremath{{\cal I}}}
\renewcommand{\L}{\ensuremath{{\cal L}}}
\newcommand{\bB}{\ensuremath{{\mathbb{B}}}}
\newcommand{\Sum}{\ensuremath{\displaystyle\sum}}
\newcommand{\centereq}[2][]{\par\noindent\relax\hfill{#2}\hfill{#1}\mbox{}\par}
\renewcommand{\qed}{$\Box$} \renewenvironment{proof} {\noindent{\bf Proof.}\ } {\hfill\qed\par\bigskip}

\begin{document}
\title{Constraint Satisfaction with Counting Quantifiers 2}
\author[barny]{Barnaby Martin}
\ead{barnabymartin@gmail.com}
\author[juraj]{Juraj Stacho}
\ead{stacho@cs.toronto.edu}
\address[barny]{School of Science and Technology, Middlesex University, The
Burroughs, Hendon, London NW4 4BT, United Kingdom}
\address[juraj]{IEOR Department, Columbia University, 500 West 120th Street,
New York, NY 10027, United States }

\begin{keyword}
constraint satisfaction problem\sep counting quantifiers\sep polynomial time\sep
Pspace-hard\sep computational complexity\sep graph homomorphism\sep graph
colouring\sep graph theory\sep cliques\sep forests
\end{keyword}

\begin{abstract}
We study \,{\em constraint satisfaction problems}\, (CSPs) in the presence of
counting quantifiers $\exists^{\geq j}$, asserting the existence of $j$ distinct
witnesses for the variable in question. As a continuation of our previous (CSR
2012) paper \cite{csr2012}, we focus on the complexity of undirected graph
templates.  As our main contribution, we settle the two principal open questions
proposed in \cite{csr2012}.  Firstly, we complete the classification of clique
templates by proving a full trichotomy for all possible combinations of counting
quantifiers and clique sizes, placing each case either in P, NP-complete or
Pspace-complete. This involves resolution of the cases in which we have the
single quantifier $\exists^{\geq j}$ on the clique $\bK_{2j}$.  Secondly, we
confirm a conjecture from \cite{csr2012}, which proposes a full dichotomy for
$\exists$ and $\exists^{\geq 2}$ on all finite undirected graphs.

The main thrust of this second result is the solution of the complexity for the
infinite path which we prove is a polynomial-time solvable problem.  By
adapting the algorithm for the infinite path we are then able to solve the
problem for finite paths, and then trees and forests. Thus as a corollary to
this work, combining with the other cases from \cite{csr2012}, we obtain  a full
dichotomy for $\exists$ and $\exists^{\geq 2}$ quantifiers on finite graphs,
each such problem being either in P or NP-hard.  Finally, we persevere with the
work of \cite{csr2012} in exploring cases in which there is dichotomy between P
and Pspace-complete, in contrast with situations in which the intermediate
NP-completeness may appear.
\end{abstract} 

\maketitle

\section{Introduction} 
The \emph{constraint satisfaction problem} CSP$(\bB)$, much studied in
artificial intelligence, is known to admit several equivalent formulations, two
of the best known of which are the query evaluation of primitive positive (pp)
sentences -- those involving only existential quantification and conjunction --
on $\bB$, and the homomorphism problem to $\bB$ (see, e.g.,
\cite{KolaitisVardiBook05}).  The problem CSP$(\bB)$ is NP-complete in general,
and a great deal of effort has been expended in classifying its complexity for
certain restricted cases. Notably it is conjectured \cite{FederVardi,JBK} that
for all fixed ${\bB}$, the problem CSP$(\bB)$ is in P or NP-complete. While this
has not been settled in general, a number of partial results are known -- e.g.
over structures of size at most three \cite{Schaefer,Bulatov} and over smooth
digraphs \cite{HellNesetril,barto:1782}.  A popular generalization of the CSP
involves considering the query evaluation problem for \emph{positive Horn} logic
-- involving only the two quantifiers, $\exists$ and $\forall$, together with
conjunction. The resulting \emph{quantified constraint satisfaction problems}
QCSP$(\bB)$ allow for a broader class, used in artificial intelligence to
capture non-monotonic reasoning, whose complexities rise to Pspace-completeness.

In this paper, we continue the project begun in \cite{csr2012} to study counting
quantifiers of the form $\exists^{\geq j}$, which allow one to assert the
existence of at least $j$ elements such that the ensuing property holds. Thus on
a structure $\bB$ with domain of size $n$, the quantifiers $\exists^{\geq 1}$
and $\exists^{\geq n}$ are precisely $\exists$ and $\forall$, respectively.

We study variants of CSP$(\bB)$ in which the input sentence to be evaluated on
$\bB$ (of size $|B|$) remains positive conjunctive in its quantifier-free part,
but is quantified by various counting quantifiers.

For $X \subseteq \{1,\ldots,|B|\}$, $X\neq \emptyset$, the $X$-CSP$(\bB)$ takes
as input a sentence given by a conjunction of atoms quantified by quantifiers of
the form $\exists^{\geq j}$ for $j \in X$. It then asks whether this sentence is
true on $\bB$.

In \cite{csr2012}, it was shown that $X$-CSP$(\bB)$ exhibits trichotomy as $\bB$
ranges over undirected, irreflexive cycles, with each problem being in either L,
NP-complete or Pspace-complete. The following classification was given for
cliques.

\begin{theorem}{\rm\cite{csr2012}}~\label{thm:cliques}
For $n\in \mathbb{N}$ and $X \subseteq \{1,\ldots,n\}$:\smallskip
\begin{compactenum}[(i)]
\item
$X$-CSP$(\bK_n)$ is in L if $n \leq 2$ or $X\cap \big\{1,\ldots,\lfloor n/2
\rfloor\big\}=\emptyset$.
\item
$X$-CSP$(\bK_n)$ is NP-complete if $n>2$ and $X=\{1\}$.
\item
$X$-CSP$(\bK_n)$ is Pspace-complete if $n>2$ and either $j \in X$ for $1<j <
n/2$ or $\{1,j\} \subseteq X$ for $j\in\big\{\lceil n/2\rceil,\ldots
,n\big\}$.
\end{compactenum}
\end{theorem}
\noindent Precisely the cases $\{j\}$-CSP$(\bK_{2j})$ are left open here. Of course,
$\{1\}$-CSP$(\bK_{2})$ is \emph{graph $2$-colorability} and is in L, but for
$j>1$ the situation was very unclear, and the referees noted specifically this
lacuna. 
\medskip

In this paper we {\bf settle this question}, and find the surprising situation
that $\{2\}$-CSP$(\bK_{4})$ is in P while $\{j\}$-CSP$(\bK_{2j})$ is
Pspace-complete for $j \geq 3$. The algorithm for the case
$\{2\}$-CSP$(\bK_{4})$ is specialized and non-trivial, and consists in iteratively
constructing a collection of forcing triples where we proceed to look for a
contradiction. 

As a second focus of the paper, we continue the study of
$\{1,2\}$-CSP($\bH$).  In particular, we focus on finite undirected
graphs for which a dichotomy was proposed in \cite{csr2012}.  As a fundamental
step towards this, we first investigate the complexity of
$\{1,2\}$-CSP$(\bP_\infty)$, where $\bP_\infty$ denotes the infinite undirected
path. We find tractability here in describing a particular unique obstruction,
which takes the form of a special walk, whose presence or absence yields the
answer to the problem. Again the algorithm is specialized and non-trivial, and
in carefully augmenting it, we construct another polynomial-time algorithm, this
time for all finite paths.

\begin{theorem}
\label{thm:finite-paths}
$\{1,2\}$-CSP{\rm($\bP_n$)} is in P, for each $n \in \mathbb{N}$.
\end{theorem}

A corollary of this is the following {\bf key result}.

\begin{corollary}
\label{cor:finite-forests}
$\{1,2\}$-CSP{\rm($\bH$)} is in P, for each forest $\bH$.
\end{corollary}

Combined with the results from \cite{HellNesetril,csr2012}, this allows us to
observe a dichotomy for $\{1,2\}$-CSP($\bH$) as $\bH$ ranges
over undirected graphs, each problem being either in P or NP-hard, in turn {\bf
settling a conjecture} proposed in \cite{csr2012}.

\begin{corollary}\label{cor:12csp-dichotomy}Let $H$ be a graph.
\begin{compactenum}[(i)]
\item $\{1,2\}$-CSP($H$) is in P if $H$ is a forest or is
bipartite with a $4$-cycle,
\item $\{1,2\}$-CSP($H$) is NP-hard in all other cases.
\end{compactenum}
\end{corollary}
In \cite{csr2012}, the main preoccupation was in the distinction between P and
NP-hard. Here we concentrate our observations to show situations in which we
have sharp dichotomies between P and Pspace-complete, as well as cases in which
NP-completeness manifests. This allows us to generalize the above as follows.

\begin{theorem}\label{thm:bip-dichotomy}Let $H$ be a bipartite graph.
\begin{compactenum}[(i)]
\item $\{1,2\}$-CSP($H$) is in P if $H$ is a forest or contains
a $4$-cycle,
\item $\{1,2\}$-CSP($H$) is Pspace-complete in all other cases.
\end{compactenum}
\end{theorem}

Taken together, our work seems to indicate a rich and largely uncharted
complexity landscape that these types of problems constitute. The associated
combinatorics to this landscape appears quite complex and the absence of simple
algebraic approach is telling. We will return to the question of algebra in the
final remarks of the paper.

The paper is structured as follows. In section \ref{sec:2csp-k4}, we describe a
characterization and a polynomial time algorithm for $\{2\}$-CSP($\bK_4$). In
section \ref{sec:ncsp-2n}, we show Pspace-hardness for $\{n\}$-CSP($\bK_{2n}$) for
$n\geq 3$. In section \ref{sec:12csp-path}, we characterize $\{1,2\}$-CSP for
the infinite path $\P_\infty$ and describe the resulting polynomial algorithm.
Then, in section \ref{sec:finite-path}, we generalize this to finite paths and
prove Theorem \ref{thm:finite-paths} and associated corollaries, in sections
\ref{sec:proof-cor3} and \ref{sec:proof-cor4}. Subsequently, in section
\ref{sec:pspace-dich}, we discuss the P/Pspace-complete dichotomy of bipartite
graphs, under $\{1,2\}$-CSP, as well as situations in which the intermediate
NP-completeness arises, in sections \ref{sec:domin} and \ref{sec:small}. We
conclude the paper in section \ref{sec:final} by giving some final thoughts.
 
\subsection{Preliminaries}

Our proofs use the game characterization and structural interpretation from
\cite{csr2012}.  For completeness, we summarize it here.  This is as follows.

Given an input $\Psi$ for $X$-CSP($\bB$), we define the following game
$\mathscr{G}(\Psi,\bB)$:

\begin{definition}
Let $\Psi:=Q_1 x_1
Q_2 x_2 \ldots Q_m x_m \ \psi(x_1,x_2,\ldots,x_m)$. Working from the outside in,
coming to a quantified variable $\exists^{\geq j} x$, the \underline{\em Prover} (female)
picks a subset $B_x$ of $j$ elements of $B$ as witnesses for $x$, and an
\uline{\em Adversary} (male) chooses one of these, say $b_x$, to be the value of $x$,
denoted by $f(x)$.
\end{definition}

\noindent Prover wins if $f$ is a homomorphism to $\bB$, i.e., if $\bB \models \psi(f(x_1),f(x_2),\ldots,f(x_m))$.

\begin{lemma}\label{lem:game}
Prover has a winning strategy in the game $\mathscr{G}(\Psi,\bB)$ iff $\bB
\models \Psi$.
\end{lemma}

\begin{definition}\label{def:instance}
Let $H$ be a graph. For an instance $\Psi$ of $X$-CSP($H$):\smallskip

\begin{compactitem}
\item define $\D_\psi$ to be
the graph whose vertices are the variables of $\Psi$ and edges are between
variables $v_i,v_j$ for which $E(v_i,v_j)$ appears in $\Psi$.\smallskip

\item denote $\prec$ the total order of variables of $\Psi$ as they are
quantified in the formula (from left to right).
\end{compactitem}
\end{definition}

\section{Algorithm for $\{2\}$-CSP($\bK_4$)}\label{sec:2csp-k4}

\begin{theorem}\label{thm:2csp-k4}
$\{2\}$-CSP($\mathbb{K}_4$) is decidable in polynomial time.
\end{theorem}

The template $\mathbb{K}_4$ has vertices $\{1,2,3,4\}$ and all possible edges
between distinct vertices.  Consider the instance $\Psi$ of
$\{2\}$-CSP($\mathbb{K}_4$) as a graph $G=\D_\psi$ together with a linear
ordering $\prec $ on $V(G)$ (see Definition~\ref{def:instance}).

We iteratively construct the following three sets: $R^+$, $R^-$, and $F$.  The
set $F$ will be a collection of unordered pairs of vertices of $G$, while $R^+$
and $R^-$ will consist of unordered triples of vertices.
(For simplicity we write $xy\in F$ in place of $\{x,y\}\in F$, and
write $xyz\in R^+$ or $R^-$ in place of $\{x,y,z\}\in R^+$ or $R^-$.)
\smallskip

We start by initializing the sets as follows: $F=E(G)$ and $R^+=R^-=\emptyset$.
Then we perform the following rules as long as possible:\smallskip

\begin{compactenum}[(X1)]
\item
If there are $x,y,z\in V(G)$ such that $\{x,y\}<z$ where $xz,yz\in F$, then add
$xyz$ into $R^-$.\smallskip

\item
If there are vertices $x,y,w,z\in V(G)$ such that $\{x,y,w\}<z$ with $wz\in F$
and $xyz\in R^-$, then add $xyw$ into $R^+$.\smallskip

\item
If there are $x,y,w,z\in V(G)$ such that $\{x,y,w\}<z$ with $wz\in F$ and
$xyz\in R^+$, then  if $\{x,y\}<w$, then add $xyw$ into $R^-$,
else add $xw$ and $yw$ into $F$.\smallskip

\item
If there are vertices $x,y,w,z\in V(G)$ such that $\{x,w\}<y<z$ with $xyz\in
R^+$ and $wyz\in R^-$, then add $xw$ into $F$, and add $xyw$ into
$R^+$.\smallskip

\item
If there are vertices $x,y,w,z\in V(G)$ such that $\{x,y,w\}<z$ where either
$xyz,wyz\in R^+$, or $xyz,wyz\in R^-$, then add $xyw$ into $R^+$.\smallskip

\item
If there are vertices $x,y,q,w,z\in V(G)$ such that $\{x,y,w\}<q<z$ where either
$xyz,wqz\in R^+$, or $xyz,wqz\in R^-$, then add $xyw$ and $xyq$ into
$R^+$.\smallskip

\item
If there are vertices $x,y,q,w,z\in V(G)$ such that $\{x,y,w\}<q<z$ where either
$xyz\in R^+$ and $wqz\in R^-$, or $xyz\in R^-$ and $wqz\in R^+$, then add $xyq$
into $R^-$, and if $\{x,y\}<w$, also add $xyw$ into $R^-$,
else add $xw$ and $yw$ into $F$.
\end{compactenum}

\begin{figure}[h!t]
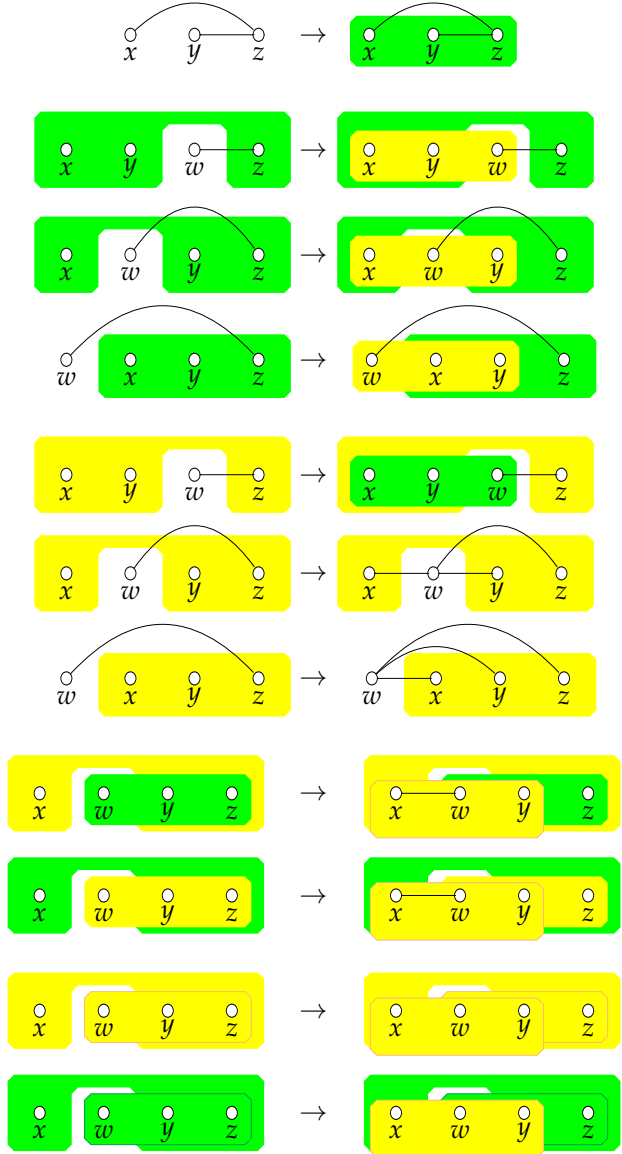

\centering
\parskip 6ex

$\xy/r2pc/:
(0,0)*[o][F]{\phantom{s}}="x";
(1,0)*[o][F]{\phantom{s}}="y";
(2,0)*[o][F]{\phantom{s}}="z";
{\ar@{-}@/^1pc/ "x";"z"};
{\ar@{-} "y";"z"};
"x"+(0,-0.3)*{x};
"y"+(0,-0.3)*{y};
"z"+(0,-0.3)*{z};
\endxy$
\quad
$\rightarrow$
\quad
$\xy/r2pc/:
@i@={(-0.3,-0.4),(-0.3,0.2),(-0.2,0.3),(2.2,0.3),(2.3,0.2),(2.3,-0.4),(2.2,-0.5),(-0.2,-0.5)},0*[green]\xypolyline{*};
(0,0)*[o][F**:white]{\phantom{s}}="x";
(1,0)*[o][F**:white]{\phantom{s}}="y";
(2,0)*[o][F**:white]{\phantom{s}}="z";
{\ar@{-}@/^1pc/ "x";"z"};
{\ar@{-} "y";"z"};
"x"+(0,-0.3)*{x};
"y"+(0,-0.3)*{y};
"z"+(0,-0.3)*{z};
\endxy$

$\xy/r2pc/:
@i@={(-0.5,-0.5),(-0.5,0.5),(-0.4,0.6),(3.4,0.6),(3.5,0.5),(3.5,-0.5),(3.4,-0.6),(2.6,-0.6),(2.5,-0.5),(2.5,0.3),(2.4,0.4),(1.6,0.4),(1.5,0.3),(1.5,-0.5),(1.4,-0.6),(-0.4,-0.6)},0*[green]\xypolyline{*};
(0,0)*[o][F**:white]{\phantom{s}}="x";
(1,0)*[o][F**:white]{\phantom{s}}="y";
(2,0)*[o][F**:white]{\phantom{s}}="w";
(3,0)*[o][F**:white]{\phantom{s}}="z";
"x"+(0,-0.3)*{x};
"y"+(0,-0.3)*{y};
"z"+(0,-0.3)*{z};
"w"+(0,-0.3)*{w};
{\ar@{-} "w";"z"};
\endxy$
\quad
$\rightarrow$
\quad
$\xy/r2pc/:
@i@={(-0.5,-0.5),(-0.5,0.5),(-0.4,0.6),(3.4,0.6),(3.5,0.5),(3.5,-0.5),(3.4,-0.6),(2.6,-0.6),(2.5,-0.5),(2.5,0.3),(2.4,0.4),(1.6,0.4),(1.5,0.3),(1.5,-0.5),(1.4,-0.6),(-0.4,-0.6)},0*[green]\xypolyline{*};
@i@={(-0.3,-0.4),(-0.3,0.2),(-0.2,0.3),(2.2,0.3),(2.3,0.2),(2.3,-0.4),(2.2,-0.5),(-0.2,-0.5)},0*[yellow]\xypolyline{*};
(0,0)*[o][F**:white]{\phantom{s}}="x";
(1,0)*[o][F**:white]{\phantom{s}}="y";
(2,0)*[o][F**:white]{\phantom{s}}="w";
(3,0)*[o][F**:white]{\phantom{s}}="z";
"x"+(0,-0.3)*{x};
"y"+(0,-0.3)*{y};
"z"+(0,-0.3)*{z};
"w"+(0,-0.3)*{w};
{\ar@{-} "w";"z"};
\endxy$\\[2ex]
$\xy/r2pc/:
@i@={(-0.5,-0.5),(-0.5,0.5),(-0.4,0.6),(3.4,0.6),(3.5,0.5),(3.5,-0.5),(3.4,-0.6),(1.6,-0.6),(1.5,-0.5),(1.5,0.3),(1.4,0.4),(0.6,0.4),(0.5,0.3),(0.5,-0.5),(0.4,-0.6),(-0.4,-0.6)},0*[green]\xypolyline{*};
(0,0)*[o][F**:white]{\phantom{s}}="x";
(1,0)*[o][F**:white]{\phantom{s}}="w";
(2,0)*[o][F**:white]{\phantom{s}}="y";
(3,0)*[o][F**:white]{\phantom{s}}="z";
"x"+(0,-0.3)*{x};
"y"+(0,-0.3)*{y};
"z"+(0,-0.3)*{z};
"w"+(0,-0.3)*{w};
{\ar@{-}@/^1.5pc/ "w";"z"};
\endxy$
\quad
$\rightarrow$
\quad
$\xy/r2pc/:
@i@={(-0.5,-0.5),(-0.5,0.5),(-0.4,0.6),(3.4,0.6),(3.5,0.5),(3.5,-0.5),(3.4,-0.6),(1.6,-0.6),(1.5,-0.5),(1.5,0.3),(1.4,0.4),(0.6,0.4),(0.5,0.3),(0.5,-0.5),(0.4,-0.6),(-0.4,-0.6)},0*[green]\xypolyline{*};
@i@={(-0.3,-0.4),(-0.3,0.2),(-0.2,0.3),(2.2,0.3),(2.3,0.2),(2.3,-0.4),(2.2,-0.5),(-0.2,-0.5)},0*[yellow]\xypolyline{*};
(0,0)*[o][F**:white]{\phantom{s}}="x";
(1,0)*[o][F**:white]{\phantom{s}}="w";
(2,0)*[o][F**:white]{\phantom{s}}="y";
(3,0)*[o][F**:white]{\phantom{s}}="z";
"x"+(0,-0.3)*{x};
"y"+(0,-0.3)*{y};
"z"+(0,-0.3)*{z};
"w"+(0,-0.3)*{w};
{\ar@{-}@/^1.5pc/ "w";"z"};
\endxy$\\[1.5ex]
$\xy/r2pc/:
@i@={(3.4,0.4),(3.5,0.3),(3.5,-0.5),(3.4,-0.6),(0.6,-0.6),(0.5,-0.5),(0.5,0.3),(0.6,0.4)},0*[green]\xypolyline{*};
(0,0)*[o][F**:white]{\phantom{s}}="w";
(1,0)*[o][F**:white]{\phantom{s}}="x";
(2,0)*[o][F**:white]{\phantom{s}}="y";
(3,0)*[o][F**:white]{\phantom{s}}="z";
"x"+(0,-0.3)*{x};
"y"+(0,-0.3)*{y};
"z"+(0,-0.3)*{z};
"w"+(0,-0.3)*{w};
{\ar@{-}@/^1.7pc/ "w";"z"};
\endxy$
\quad
$\rightarrow$
\quad
$\xy/r2pc/:
@i@={(3.4,0.4),(3.5,0.3),(3.5,-0.5),(3.4,-0.6),(0.6,-0.6),(0.5,-0.5),(0.5,0.3),(0.6,0.4)},0*[green]\xypolyline{*};
@i@={(-0.3,-0.4),(-0.3,0.2),(-0.2,0.3),(2.2,0.3),(2.3,0.2),(2.3,-0.4),(2.2,-0.5),(-0.2,-0.5)},0*[yellow]\xypolyline{*};
(0,0)*[o][F**:white]{\phantom{s}}="w";
(1,0)*[o][F**:white]{\phantom{s}}="x";
(2,0)*[o][F**:white]{\phantom{s}}="y";
(3,0)*[o][F**:white]{\phantom{s}}="z";
"x"+(0,-0.3)*{x};
"y"+(0,-0.3)*{y};
"z"+(0,-0.3)*{z};
"w"+(0,-0.3)*{w};
{\ar@{-}@/^1.7pc/ "w";"z"};
\endxy$

$\xy/r2pc/:
@i@={(-0.5,-0.5),(-0.5,0.5),(-0.4,0.6),(3.4,0.6),(3.5,0.5),(3.5,-0.5),(3.4,-0.6),(2.6,-0.6),(2.5,-0.5),(2.5,0.3),(2.4,0.4),(1.6,0.4),(1.5,0.3),(1.5,-0.5),(1.4,-0.6),(-0.4,-0.6)},0*[yellow]\xypolyline{*};
(0,0)*[o][F**:white]{\phantom{s}}="x";
(1,0)*[o][F**:white]{\phantom{s}}="y";
(2,0)*[o][F**:white]{\phantom{s}}="w";
(3,0)*[o][F**:white]{\phantom{s}}="z";
"x"+(0,-0.3)*{x};
"y"+(0,-0.3)*{y};
"z"+(0,-0.3)*{z};
"w"+(0,-0.3)*{w};
{\ar@{-} "w";"z"};
\endxy$
\quad
$\rightarrow$
\quad
$\xy/r2pc/:
@i@={(-0.5,-0.5),(-0.5,0.5),(-0.4,0.6),(3.4,0.6),(3.5,0.5),(3.5,-0.5),(3.4,-0.6),(2.6,-0.6),(2.5,-0.5),(2.5,0.3),(2.4,0.4),(1.6,0.4),(1.5,0.3),(1.5,-0.5),(1.4,-0.6),(-0.4,-0.6)},0*[yellow]\xypolyline{*};
@i@={(-0.3,-0.4),(-0.3,0.2),(-0.2,0.3),(2.2,0.3),(2.3,0.2),(2.3,-0.4),(2.2,-0.5),(-0.2,-0.5)},0*[green]\xypolyline{*};
(0,0)*[o][F**:white]{\phantom{s}}="x";
(1,0)*[o][F**:white]{\phantom{s}}="y";
(2,0)*[o][F**:white]{\phantom{s}}="w";
(3,0)*[o][F**:white]{\phantom{s}}="z";
"x"+(0,-0.3)*{x};
"y"+(0,-0.3)*{y};
"z"+(0,-0.3)*{z};
"w"+(0,-0.3)*{w};
{\ar@{-} "w";"z"};
\endxy$\\[1.5ex]
$\xy/r2pc/:
@i@={(-0.5,-0.5),(-0.5,0.5),(-0.4,0.6),(3.4,0.6),(3.5,0.5),(3.5,-0.5),(3.4,-0.6),(1.6,-0.6),(1.5,-0.5),(1.5,0.3),(1.4,0.4),(0.6,0.4),(0.5,0.3),(0.5,-0.5),(0.4,-0.6),(-0.4,-0.6)},0*[yellow]\xypolyline{*};
(0,0)*[o][F**:white]{\phantom{s}}="x";
(1,0)*[o][F**:white]{\phantom{s}}="w";
(2,0)*[o][F**:white]{\phantom{s}}="y";
(3,0)*[o][F**:white]{\phantom{s}}="z";
"x"+(0,-0.3)*{x};
"y"+(0,-0.3)*{y};
"z"+(0,-0.3)*{z};
"w"+(0,-0.3)*{w};
{\ar@{-}@/^1.5pc/ "w";"z"};
\endxy$
\quad
$\rightarrow$
\quad
$\xy/r2pc/:
@i@={(-0.5,-0.5),(-0.5,0.5),(-0.4,0.6),(3.4,0.6),(3.5,0.5),(3.5,-0.5),(3.4,-0.6),(1.6,-0.6),(1.5,-0.5),(1.5,0.3),(1.4,0.4),(0.6,0.4),(0.5,0.3),(0.5,-0.5),(0.4,-0.6),(-0.4,-0.6)},0*[yellow]\xypolyline{*};
(0,0)*[o][F**:white]{\phantom{s}}="x";
(1,0)*[o][F**:white]{\phantom{s}}="w";
(2,0)*[o][F**:white]{\phantom{s}}="y";
(3,0)*[o][F**:white]{\phantom{s}}="z";
"x"+(0,-0.3)*{x};
"y"+(0,-0.3)*{y};
"z"+(0,-0.3)*{z};
"w"+(0,-0.3)*{w};
{\ar@{-}@/^1.5pc/ "w";"z"};
{\ar@{-} "w";"x"};
{\ar@{-} "w";"y"};
\endxy$\\[1.5ex]
$\xy/r2pc/:
@i@={(3.4,0.4),(3.5,0.3),(3.5,-0.5),(3.4,-0.6),(0.6,-0.6),(0.5,-0.5),(0.5,0.3),(0.6,0.4)},0*[yellow]\xypolyline{*};
(0,0)*[o][F**:white]{\phantom{s}}="w";
(1,0)*[o][F**:white]{\phantom{s}}="x";
(2,0)*[o][F**:white]{\phantom{s}}="y";
(3,0)*[o][F**:white]{\phantom{s}}="z";
"x"+(0,-0.3)*{x};
"y"+(0,-0.3)*{y};
"z"+(0,-0.3)*{z};
"w"+(0,-0.3)*{w};
{\ar@{-}@/^1.7pc/ "w";"z"};
\endxy$
\quad
$\rightarrow$
\quad
$\xy/r2pc/:
@i@={(3.4,0.4),(3.5,0.3),(3.5,-0.5),(3.4,-0.6),(0.6,-0.6),(0.5,-0.5),(0.5,0.3),(0.6,0.4)},0*[yellow]\xypolyline{*};
(0,0)*[o][F**:white]{\phantom{s}}="w";
(1,0)*[o][F**:white]{\phantom{s}}="x";
(2,0)*[o][F**:white]{\phantom{s}}="y";
(3,0)*[o][F**:white]{\phantom{s}}="z";
"x"+(0,-0.3)*{x};
"y"+(0,-0.3)*{y};
"z"+(0,-0.3)*{z};
"w"+(0,-0.3)*{w};
{\ar@{-}@/^1.7pc/ "w";"z"};
{\ar@{-} "w";"x"};
{\ar@{-}@/^1pc/ "w";"y"};
\endxy$

$\xy/r2pc/:
@i@={(-0.5,-0.5),(-0.5,0.5),(-0.4,0.6),(3.4,0.6),(3.5,0.5),(3.5,-0.5),(3.4,-0.6),(1.6,-0.6),(1.5,-0.5),(1.5,0.3),(1.4,0.4),(0.6,0.4),(0.5,0.3),(0.5,-0.5),(0.4,-0.6),(-0.4,-0.6)},0*[yellow]\xypolyline{*};
@i@={(0.7,-0.4),(0.7,0.2),(0.8,0.3),(3.2,0.3),(3.3,0.2),(3.3,-0.4),(3.2,-0.5),(0.8,-0.5)},0*[green]\xypolyline{*};
(0,0)*[o][F**:white]{\phantom{s}}="x";
(1,0)*[o][F**:white]{\phantom{s}}="w";
(2,0)*[o][F**:white]{\phantom{s}}="y";
(3,0)*[o][F**:white]{\phantom{s}}="z";
"x"+(0,-0.3)*{x};
"y"+(0,-0.3)*{y};
"z"+(0,-0.3)*{z};
"w"+(0,-0.3)*{w};
\endxy$
\qquad
$\rightarrow$
\qquad
$\xy/r2pc/:
@i@={(-0.5,-0.5),(-0.5,0.5),(-0.4,0.6),(3.4,0.6),(3.5,0.5),(3.5,-0.5),(3.4,-0.6),(1.6,-0.6),(1.5,-0.5),(1.5,0.3),(1.4,0.4),(0.6,0.4),(0.5,0.3),(0.5,-0.5),(0.4,-0.6),(-0.4,-0.6)},0*[yellow]\xypolyline{*};
@i@={(0.7,-0.4),(0.7,0.2),(0.8,0.3),(3.2,0.3),(3.3,0.2),(3.3,-0.4),(3.2,-0.5),(0.8,-0.5)},0*[green]\xypolyline{*};
@i@={(-0.4,-0.6),(-0.4,0.1),(-0.3,0.2),(2.2,0.2),(2.3,0.1),(2.3,-0.6),(2.2,-0.7),(-0.3,-0.7)},0*[yellow]\xypolyline{*};
@i@={(-0.4,-0.6),(-0.4,0.1),(-0.3,0.2),(2.2,0.2),(2.3,0.1),(2.3,-0.6),(2.2,-0.7),(-0.3,-0.7),(-0.4,-0.6)},0*[Apricot]\xypolyline{};
(0,0)*[o][F**:white]{\phantom{s}}="x";
(1,0)*[o][F**:white]{\phantom{s}}="w";
(2,0)*[o][F**:white]{\phantom{s}}="y";
(3,0)*[o][F**:white]{\phantom{s}}="z";
"x"+(0,-0.3)*{x};
"y"+(0,-0.3)*{y};
"z"+(0,-0.3)*{z};
"w"+(0,-0.3)*{w};
{\ar@{-} "x";"w"};
\endxy$\\[5ex]
$\xy/r2pc/:
@i@={(-0.5,-0.5),(-0.5,0.5),(-0.4,0.6),(3.4,0.6),(3.5,0.5),(3.5,-0.5),(3.4,-0.6),(1.6,-0.6),(1.5,-0.5),(1.5,0.3),(1.4,0.4),(0.6,0.4),(0.5,0.3),(0.5,-0.5),(0.4,-0.6),(-0.4,-0.6)},0*[green]\xypolyline{*};
@i@={(0.7,-0.4),(0.7,0.2),(0.8,0.3),(3.2,0.3),(3.3,0.2),(3.3,-0.4),(3.2,-0.5),(0.8,-0.5)},0*[yellow]\xypolyline{*};
(0,0)*[o][F**:white]{\phantom{s}}="x";
(1,0)*[o][F**:white]{\phantom{s}}="w";
(2,0)*[o][F**:white]{\phantom{s}}="y";
(3,0)*[o][F**:white]{\phantom{s}}="z";
"x"+(0,-0.3)*{x};
"y"+(0,-0.3)*{y};
"z"+(0,-0.3)*{z};
"w"+(0,-0.3)*{w};
\endxy$
\qquad
$\rightarrow$
\qquad
$\xy/r2pc/:
@i@={(-0.5,-0.5),(-0.5,0.5),(-0.4,0.6),(3.4,0.6),(3.5,0.5),(3.5,-0.5),(3.4,-0.6),(1.6,-0.6),(1.5,-0.5),(1.5,0.3),(1.4,0.4),(0.6,0.4),(0.5,0.3),(0.5,-0.5),(0.4,-0.6),(-0.4,-0.6)},0*[green]\xypolyline{*};
@i@={(0.7,-0.4),(0.7,0.2),(0.8,0.3),(3.2,0.3),(3.3,0.2),(3.3,-0.4),(3.2,-0.5),(0.8,-0.5)},0*[yellow]\xypolyline{*};
@i@={(-0.4,-0.6),(-0.4,0.1),(-0.3,0.2),(2.2,0.2),(2.3,0.1),(2.3,-0.6),(2.2,-0.7),(-0.3,-0.7)},0*[yellow]\xypolyline{*};
@i@={(-0.4,-0.6),(-0.4,0.1),(-0.3,0.2),(2.2,0.2),(2.3,0.1),(2.3,-0.6),(2.2,-0.7),(-0.3,-0.7),(-0.4,-0.6)},0*[Apricot]\xypolyline{};
(0,0)*[o][F**:white]{\phantom{s}}="x";
(1,0)*[o][F**:white]{\phantom{s}}="w";
(2,0)*[o][F**:white]{\phantom{s}}="y";
(3,0)*[o][F**:white]{\phantom{s}}="z";
"x"+(0,-0.3)*{x};
"y"+(0,-0.3)*{y};
"z"+(0,-0.3)*{z};
"w"+(0,-0.3)*{w};
{\ar@{-} "x";"w"};
\endxy$

$\xy/r2pc/:
@i@={(-0.5,-0.5),(-0.5,0.5),(-0.4,0.6),(3.4,0.6),(3.5,0.5),(3.5,-0.5),(3.4,-0.6),(1.6,-0.6),(1.5,-0.5),(1.5,0.3),(1.4,0.4),(0.6,0.4),(0.5,0.3),(0.5,-0.5),(0.4,-0.6),(-0.4,-0.6)},0*[yellow]\xypolyline{*};
@i@={(0.7,-0.4),(0.7,0.2),(0.8,0.3),(3.2,0.3),(3.3,0.2),(3.3,-0.4),(3.2,-0.5),(0.8,-0.5)},0*[yellow]\xypolyline{*};
@i@={(0.7,-0.4),(0.7,0.2),(0.8,0.3),(3.2,0.3),(3.3,0.2),(3.3,-0.4),(3.2,-0.5),(0.8,-0.5),(0.7,-0.4)},0*[Apricot]\xypolyline{};
(0,0)*[o][F**:white]{\phantom{s}}="x";
(1,0)*[o][F**:white]{\phantom{s}}="w";
(2,0)*[o][F**:white]{\phantom{s}}="y";
(3,0)*[o][F**:white]{\phantom{s}}="z";
"x"+(0,-0.3)*{x};
"y"+(0,-0.3)*{y};
"z"+(0,-0.3)*{z};
"w"+(0,-0.3)*{w};
\endxy$
\qquad
$\rightarrow$
\qquad
$\xy/r2pc/:
@i@={(-0.5,-0.5),(-0.5,0.5),(-0.4,0.6),(3.4,0.6),(3.5,0.5),(3.5,-0.5),(3.4,-0.6),(1.6,-0.6),(1.5,-0.5),(1.5,0.3),(1.4,0.4),(0.6,0.4),(0.5,0.3),(0.5,-0.5),(0.4,-0.6),(-0.4,-0.6)},0*[yellow]\xypolyline{*};
@i@={(0.7,-0.4),(0.7,0.2),(0.8,0.3),(3.2,0.3),(3.3,0.2),(3.3,-0.4),(3.2,-0.5),(0.8,-0.5)},0*[yellow]\xypolyline{*};
@i@={(0.7,-0.4),(0.7,0.2),(0.8,0.3),(3.2,0.3),(3.3,0.2),(3.3,-0.4),(3.2,-0.5),(0.8,-0.5),(0.7,-0.4)},0*[Apricot]\xypolyline{};
@i@={(-0.4,-0.6),(-0.4,0.1),(-0.3,0.2),(2.2,0.2),(2.3,0.1),(2.3,-0.6),(2.2,-0.7),(-0.3,-0.7)},0*[yellow]\xypolyline{*};
@i@={(-0.4,-0.6),(-0.4,0.1),(-0.3,0.2),(2.2,0.2),(2.3,0.1),(2.3,-0.6),(2.2,-0.7),(-0.3,-0.7),(-0.4,-0.6)},0*[Apricot]\xypolyline{};
(0,0)*[o][F**:white]{\phantom{s}}="x";
(1,0)*[o][F**:white]{\phantom{s}}="w";
(2,0)*[o][F**:white]{\phantom{s}}="y";
(3,0)*[o][F**:white]{\phantom{s}}="z";
"x"+(0,-0.3)*{x};
"y"+(0,-0.3)*{y};
"z"+(0,-0.3)*{z};
"w"+(0,-0.3)*{w};
\endxy$\\[5ex]
$\xy/r2pc/:
@i@={(-0.5,-0.5),(-0.5,0.5),(-0.4,0.6),(3.4,0.6),(3.5,0.5),(3.5,-0.5),(3.4,-0.6),(1.6,-0.6),(1.5,-0.5),(1.5,0.3),(1.4,0.4),(0.6,0.4),(0.5,0.3),(0.5,-0.5),(0.4,-0.6),(-0.4,-0.6)},0*[green]\xypolyline{*};
@i@={(0.7,-0.4),(0.7,0.2),(0.8,0.3),(3.2,0.3),(3.3,0.2),(3.3,-0.4),(3.2,-0.5),(0.8,-0.5)},0*[green]\xypolyline{*};
@i@={(0.7,-0.4),(0.7,0.2),(0.8,0.3),(3.2,0.3),(3.3,0.2),(3.3,-0.4),(3.2,-0.5),(0.8,-0.5),(0.7,-0.4)},0*[ForestGreen]\xypolyline{};
(0,0)*[o][F**:white]{\phantom{s}}="x";
(1,0)*[o][F**:white]{\phantom{s}}="w";
(2,0)*[o][F**:white]{\phantom{s}}="y";
(3,0)*[o][F**:white]{\phantom{s}}="z";
"x"+(0,-0.3)*{x};
"y"+(0,-0.3)*{y};
"z"+(0,-0.3)*{z};
"w"+(0,-0.3)*{w};
\endxy$
\qquad
$\rightarrow$
\qquad
$\xy/r2pc/:
@i@={(-0.5,-0.5),(-0.5,0.5),(-0.4,0.6),(3.4,0.6),(3.5,0.5),(3.5,-0.5),(3.4,-0.6),(1.6,-0.6),(1.5,-0.5),(1.5,0.3),(1.4,0.4),(0.6,0.4),(0.5,0.3),(0.5,-0.5),(0.4,-0.6),(-0.4,-0.6)},0*[green]\xypolyline{*};
@i@={(0.7,-0.4),(0.7,0.2),(0.8,0.3),(3.2,0.3),(3.3,0.2),(3.3,-0.4),(3.2,-0.5),(0.8,-0.5)},0*[green]\xypolyline{*};
@i@={(0.7,-0.4),(0.7,0.2),(0.8,0.3),(3.2,0.3),(3.3,0.2),(3.3,-0.4),(3.2,-0.5),(0.8,-0.5),(0.7,-0.4)},0*[ForestGreen]\xypolyline{};
@i@={(-0.4,-0.6),(-0.4,0.1),(-0.3,0.2),(2.2,0.2),(2.3,0.1),(2.3,-0.6),(2.2,-0.7),(-0.3,-0.7)},0*[yellow]\xypolyline{*};
@i@={(-0.4,-0.6),(-0.4,0.1),(-0.3,0.2),(2.2,0.2),(2.3,0.1),(2.3,-0.6),(2.2,-0.7),(-0.3,-0.7),(-0.4,-0.6)},0*[Apricot]\xypolyline{};
(0,0)*[o][F**:white]{\phantom{s}}="x";
(1,0)*[o][F**:white]{\phantom{s}}="w";
(2,0)*[o][F**:white]{\phantom{s}}="y";
(3,0)*[o][F**:white]{\phantom{s}}="z";
"x"+(0,-0.3)*{x};
"y"+(0,-0.3)*{y};
"z"+(0,-0.3)*{z};
"w"+(0,-0.3)*{w};
\endxy$
\caption{Illustrating rules (X1)-(X4).\label{fig:2csp-k4-rules}}
\end{figure}

\begin{figure}
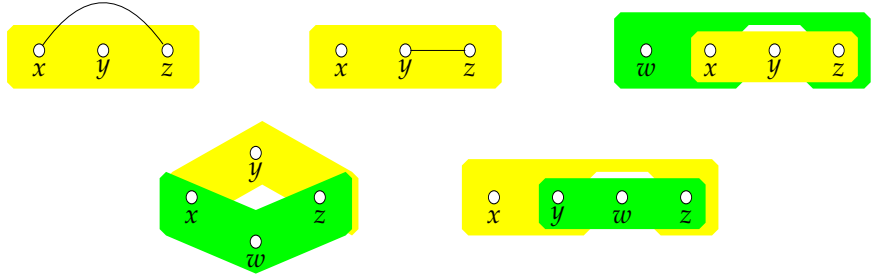
\centering
$\xy/r2pc/:
@i@={(2.4,0.4),(2.5,0.3),(2.5,-0.5),(2.4,-0.6),(-0.4,-0.6),(-0.5,-0.5),(-0.5,0.3),(-0.4,0.4)},0*[yellow]\xypolyline{*};
(0,0)*[o][F**:white]{\phantom{s}}="x";
(1,0)*[o][F**:white]{\phantom{s}}="y";
(2,0)*[o][F**:white]{\phantom{s}}="z";
"x"+(0,-0.3)*{x};
"y"+(0,-0.3)*{y};
"z"+(0,-0.3)*{z};
{\ar@{-}@/^1.5pc/ "x";"z"};
\endxy$\qquad\qquad\qquad
$\xy/r2pc/:
@i@={(2.4,0.4),(2.5,0.3),(2.5,-0.5),(2.4,-0.6),(-0.4,-0.6),(-0.5,-0.5),(-0.5,0.3),(-0.4,0.4)},0*[yellow]\xypolyline{*};
(0,0)*[o][F**:white]{\phantom{s}}="x";
(1,0)*[o][F**:white]{\phantom{s}}="y";
(2,0)*[o][F**:white]{\phantom{s}}="z";
"x"+(0,-0.3)*{x};
"y"+(0,-0.3)*{y};
"z"+(0,-0.3)*{z};
{\ar@{-} "y";"z"};
\endxy$\qquad\qquad\qquad$\xy/r2pc/:
@i@={(-0.5,-0.5),(-0.5,0.5),(-0.4,0.6),(3.4,0.6),(3.5,0.5),(3.5,-0.5),(3.4,-0.6),(2.6,-0.6),(2.5,-0.5),(2.5,0.3),(2.4,0.4),(1.6,0.4),(1.5,0.3),(1.5,-0.5),(1.4,-0.6),(-0.4,-0.6)},0*[green]\xypolyline{*};
@i@={(0.7,-0.4),(0.7,0.2),(0.8,0.3),(3.2,0.3),(3.3,0.2),(3.3,-0.4),(3.2,-0.5),(0.8,-0.5)},0*[yellow]\xypolyline{*};
(0,0)*[o][F**:white]{\phantom{s}}="w";
(1,0)*[o][F**:white]{\phantom{s}}="x";
(2,0)*[o][F**:white]{\phantom{s}}="y";
(3,0)*[o][F**:white]{\phantom{s}}="z";
"x"+(0,-0.3)*{x};
"y"+(0,-0.3)*{y};
"z"+(0,-0.3)*{z};
"w"+(0,-0.3)*{w};
\endxy$\\[5ex]
$\xy/r2pc/:
@i@={(2.5,0.4),(2.6,0.3),(2.6,-0.5),(2.5,-0.6),(1.1,0.2),(-0.3,-0.6),(-0.4,-0.5),(-0.4,0.3),(-0.3,0.4),(1.1,1.2)},0*[yellow]\xypolyline{*};
@i@={(2.4,0.4),(2.5,0.3),(2.5,-0.5),(2.4,-0.6),(1,-1.2),(-0.4,-0.6),(-0.5,-0.5),(-0.5,0.3),(-0.4,0.4),(1,-0.2)},0*[green]\xypolyline{*};
(0,0)*[o][F**:white]{\phantom{s}}="x";
(1,0.7)*[o][F**:white]{\phantom{s}}="y";
(1,-0.7)*[o][F**:white]{\phantom{s}}="w";
(2,0)*[o][F**:white]{\phantom{s}}="z";
"x"+(0,-0.3)*{x};
"y"+(0,-0.3)*{y};
"w"+(0,-0.3)*{w};
"z"+(0,-0.3)*{z};
\endxy$\qquad\qquad\qquad
$\xy/r2pc/:
@i@={(-0.5,-0.5),(-0.5,0.5),(-0.4,0.6),(3.4,0.6),(3.5,0.5),(3.5,-0.5),(3.4,-0.6),(2.6,-0.6),(2.5,-0.5),(2.5,0.3),(2.4,0.4),(1.6,0.4),(1.5,0.3),(1.5,-0.5),(1.4,-0.6),(-0.4,-0.6)},0*[yellow]\xypolyline{*};
@i@={(0.7,-0.4),(0.7,0.2),(0.8,0.3),(3.2,0.3),(3.3,0.2),(3.3,-0.4),(3.2,-0.5),(0.8,-0.5)},0*[green]\xypolyline{*};
(0,0)*[o][F**:white]{\phantom{s}}="x";
(1,0)*[o][F**:white]{\phantom{s}}="y";
(2,0)*[o][F**:white]{\phantom{s}}="w";
(3,0)*[o][F**:white]{\phantom{s}}="z";
"x"+(0,-0.3)*{x};
"y"+(0,-0.3)*{y};
"z"+(0,-0.3)*{z};
"w"+(0,-0.3)*{w};
\endxy$
\caption{Five forbidden configurations of Theorem \ref{thm:2csp-k4-char}.\label{fig:2csp-k4-forb}}
\end{figure}

\begin{theorem}\label{thm:2csp-k4-char}
The following are equivalent:
\begin{compactenum}[(i)]
\item $\mathbb{K}_4\models \Psi$\smallskip
\item Prover has a winning strategy in
$\mathscr{G}(\Psi,\mathbb{K}_4)$.\smallskip
\item Prover can play so that in every instance of the game, the resulting
mapping $f:V(G)\rightarrow\{1,2,3,4\}$ satisfies the following
properties:\smallskip
\begin{compactenum}[(S1)]
\item For every $xy\in F$, we have: $f(x)\neq f(y)$.
\item For every $xyz\in R^+$ such that $x<y<z$:
if $f(x)\neq f(y)$, then $f(z)\in\big\{f(x),f(y)\big\}$.
\item For every $xyz\in R^-$ such that $x<y<z$:
if $f(x)\neq f(y)$, then $f(z)\not\in\big\{f(x),f(y)\big\}$.
\end{compactenum}\smallskip
\item 
there is no triple $xyz$ in $R^+$ such that $x<y<z$ and (see Fig. \ref{fig:2csp-k4-forb})
\begin{compactitem}
\item $xz\in F$ or $yz\in F$,
\item or $xwz\in R^-$ for some $w<z$ (possibly $w=y$),
\item or $ywz\in R^-$ for some $y<w<z$.
\end{compactitem}
\end{compactenum}
\end{theorem}

\subsection{Proof of Theorem \ref{thm:2csp-k4-char}}

{\bf (i)$\iff$(ii)} is by definition. {\bf (iii)$\Rightarrow$(ii)} is implied by the
fact that $F\supseteq E(G)$, and that
by (iii) Prover can play to satisfy (S1). Thus in every instance of the
game the mapping $f$ is a homomorphism of $G$ to $\mathbb{K}_4$ $\Rightarrow$
(ii).
\medskip

{\bf (ii)$\Rightarrow$(iii)}:
Suppose that Prover plays a winning strategy in the game
$\mathscr{G}(\Psi,\mathbb{K}_4)$ but (iii) fails.  We show that this is
impossible. Namely, we show how Adversary can play to win.

Consider an instance of the game producing a mapping $f$. We say that (S1) fails
at a vertex $v$ if there exists $a\in V(G)$ with $a<v$ such that $av\in F$ and
$f(a)=f(v)$. We say that (S2) fails at $v$ if there exist $a,b\in V(G)$ with
$a<b<v$ such that $abv\in R^+$ while $f(a)\neq f(b)$ and
$f(v)\not\in\{f(a),f(b)\}$. We say that (S3) fails at $v$ if there exist $a,b\in
V(G)$ with $a<b<v$ such that $abv\in R^-$ while $f(a)\neq f(b)$ and
$f(v)\in\{f(a),f(b)\}$.

Since (iii) fails, there is an instance of the game producing a mapping $f$ that fails (S1)-(S3)
at some vertex $v$.  Among all such instances, pick one for which $v$ is largest possible with
respect to the order $<$. We will show that this is impossible, namely we
will produce a (possibly) different instance violating the maximality of this choice.
Note that, since we assume that Prover plays a winning strategy, the
mapping $f$ is a homomorphism of $G$ to $\mathbb{K}_4$.
\medskip

\noindent{\bf Case 1:} Suppose that (S1) fails at $v$. Then there is $a<v$
such that $av\in F$ and $f(a)=f(v)$.  If $av\in E(G)$, then the mapping $f$ is
not a homomorphism of $G$ to $\mathbb{K}_4$. Thus Adversary wins, which
contradicts (ii). So we may assume that $av\not\in E(G)$. This implies that $av$
was added to $F$ using one of the rules (X3), (X4), (X7).  \smallskip

\noindent{\bf Case 1.1:}
Suppose that $av$ was added to $F$ using (X3). Then there exist vertices
$x,y,w,z$ where $\{x,y,w\}<z$ and $\{x,y\}\not< w$ such that $wz\in F$ and
$xyz\in R^+$, and either $a\in\{x,y\}$ and $v=w$, or $v\in\{x,y\}$ and $a=w$.
In particular, since $f(a)=f(v)$, we deduce that $f(x)=f(w)$ or $f(y)=f(w)$.

We may assume by symmetry that $x<y$. Recall that $\{x,y\}\not< w$. Thus
$\{x,w\}< y$.  Consider the point of the instance of the game producing $f$ when
Prover offers values for $y$. From this point on, we have Adversary play as
follows: for $y$, if $f(x)=f(w)$, choose any value that is different from
$f(w)$; if $f(x)\neq f(w)$, choose $f(y)$ for $y$. Let $\beta$ denote the value
chosen for $y$.  Observe that the choice is always possible, since Prover offers
for $y$ two distinct values, one of which is $f(y)$.  Moreover, the choice
guarantees that $f(x)\neq\beta$, since either $f(x)=f(w)\neq\beta$ or
$f(x)\neq f(w)=f(y)=\beta$. For this, recall that $f(x)=f(w)$ or $f(y)=f(w)$.
Then for $z$, if $f(w)$ is offered by Prover, we let Adversary choose $f(w)$ for
$z$; otherwise Adversary chooses any value different from $f(x)$ and $\beta$.
Let $\alpha$ denote the value chosen for $z$.  Again, note that the choice is
always possible, in particular in the latter case where Prover offers for $z$
two distinct values, neither of which is $f(w)$, while $f(w)=f(x)$ or
$f(w)=f(y)=\beta$.  For the remaining vertices, we let Adversary play any
choices. This produces a (possibly) different instance of the game. It follows
that this instance fails (S1) or (S2) at~$z$.  Namely, if $f(w)$ was offered for
$z$, then $\alpha=f(w)$ and (S1) fails at $z$, since $wz\in F$. If $f(w)$ was not
offered for $z$, then $\alpha\not\in\{f(x),\beta\}$, in which case (S2) fails at
$z$, because $xyz\in R^+$. However, since $v<z$, this contradicts our choice of
$v$.  \smallskip

\noindent{\bf Case 1.2:}
Suppose that $av$ was added to $F$ using (X4). Then there exist vertices
$x,y,w,z$ where $\{x,w\}<y<z$ such that $xyz\in R^+$ and $wyz\in R^-$, and where
$\{x,w\}=\{a,v\}$. In particular, we deduce that $f(x)=f(w)$.

We consider the point when Prover offers values for $y$ and have Adversary play
as follows: for $y$, choose any value different from $f(x)$; let $\beta$ denote
this value. Note that $\beta\neq f(x)$. Then for $z$, choose any value, denote
it $\alpha$, and then play any choices for the remaining vertices. The instance
this produces fails (S2) or (S3) at $z$. Namely, if
$\alpha\not\in\{f(x),\beta\}$, then (S2) fails, since $xyz\in R^+$, while if
$\alpha\in\{f(x),\beta\}$, then (S3) fails, since $f(x)=f(w)$ and $wyz\in R^-$.
This again contradicts our choice of $v$, since $v<z$.  \smallskip

\noindent{\bf Case 1.3:}
Suppose that $av$ was added to $F$ using (X7).  Then there exist vertices
$x,y,q,w,z$ where $\{x,y,w\}<q<z$ and $\{x,y\}\not< w$ such that $xyz\in R^+$,
$wqz\in R^-$, or $xyz\in R^-$, $wqz\in R^+$, and such that $a\in \{x,y\}$ and
$v=w$, or $a=w$ and $v\in\{x,y\}$.  Since $f(a)=f(v)$, we deduce that
$f(x)=f(w)$ or $f(y)=f(w)$.

By symmetry, assume $x<y$. Thus $\{x,w\}<y<q<z$ because $\{x,y\}\not< w$. We
have Adversary play from $y$ as follows: for $y$, if $f(x)=f(w)$, choose any
value different from $f(w)$; if $f(x)\neq f(w)$, choose $f(y)$ for $y$. Let
$\beta$ denote the value chosen for $y$.  Note that $f(x)\neq \beta$ and
$f(w)\in\{f(x),\beta\}$, since either $f(x)=f(w)\neq\beta$ or $f(x)\neq
f(w)=f(y)=\beta$.  For this, recall that $f(x)=f(w)$ or $f(y)=f(w)$.  Next, for
$q$, we let Adversary choose any value different from $f(w)$, and denote it
$\gamma$. Note that $\gamma\neq f(w)$. Finally, for $z$, if $xyz\in R^+$ and
$wqz\in R^-$, Adversary chooses $f(w)$ if offered by Prover, and if not, he
chooses any value different from $f(x)$ and $\beta$.  Similarly, if $xyz\in R^-$
and $wqz\in R^+$, Adversary chooses $f(x)$ or $\beta$ if offered by Prover, and
otherwise he chooses any value different from $f(w)$ and $\gamma$.  Let $\alpha$
denote the value chosen for $z$.  Again, we see that this choice is always
possible, since $f(w)\in\{f(x),\beta\}$.  Adversary plays any choices for the
rest.  We claim that the instance this produces fails (S2) or (S3) at $z$.
Namely, if $xyz\in R^+$ and $wqz\in R^-$, then (S2) fails at $z$ if
$\alpha\not\in\{f(x),\beta\}$, since $xyz\in R^+$, while if
$\alpha\in\{f(x),\beta\}$, then (S3) fails at $z$, since in that case we must
have $\alpha=f(w)\neq\gamma$ by the choice of $\alpha$, while $wqz\in R^-$.
Similarly, if $xyz\in R^-$ and $wqz\in R^+$, then either
$\alpha\in\{f(x),\beta\}$ and (S3) fails at $z$, since $xyz\in R^-$, or
$\alpha\not\in\{f(x),\beta\}$ in which case $\alpha\not\in\{f(w),\gamma\}$ and
(S2) fails, since $wqz\in R^+$. Since $v<z$, this contradicts our choice of
$v$.  \medskip

From now on, we may assume that (S1) does not fail at $v$.
\medskip

\noindent{\bf Case 2:} Suppose that (S2) fails at $v$. Then there are vertices
$a,b$ with $a<b<v$ such that $abv\in R^+$ while $f(a)\neq f(b)$ and
$f(v)\not\in\{f(a),f(b)\}$.  Since the set $R^+$ is initially empty, the triple
$abv$ was added to $R^+$ by one of the rules (X2), (X4), (X5), or (X6).
\smallskip

\noindent{\bf Case 2.1:} Suppose that $abv$ was added to $R^+$ using (X2).  Then
there exist vertices $x,y,w,z$ where $\{x,y,w\}<z$ such that $wz\in F$ and
$xyz\in R^-$, and where $\{a,b,v\}=\{x,y,w\}$.  We deduce that $f(x),f(y),f(w)$
are three distinct values, since $f(a),f(b),f(v)$ are.  We have Adversary play
from $z$ as follows: for $z$, from the two offered values, one will be in
$\{f(x),f(y),f(w)\}$; choose this value $\alpha$.  For the rest, play any
choices. It follows that this instance fails (S1) or (S3) at $z$. Namely, if
$\alpha=f(w)$, then (S1) fails, since $wz\in F$, while if
$\alpha\in\{f(x),f(y)\}$, then (S3) fails, since $xyz\in R^-$.  This contradicts
our choice of $v$, since $v<z$.  \smallskip

\noindent{\bf Case 2.2:} Suppose that $abv$ was added to $R^+$ using (X4). Then
there are vertices $x,y,w,z$ where $\{x,w\}<y<z$ such that $xyz\in R^+$ and
$wyz\in R^-$, and where $\{a,b,v\}=\{x,y,w\}$. Again, we deduce that
$f(x),f(y),f(w)$ are pairwise distinct. Adversary plays from $z$ as follows:
for $z$, if $f(y)$ is offered, choose $f(y)$; otherwise, choose any value
different from $f(x)$.  For the rest, play any choices. Let $\alpha$ be the
value chosen for $z$. We claim that this instance fails (S2) or (S3) at $z$.
Namely, if $\alpha=f(y)$, then (S3) fails, since $f(w)\neq f(y)$ and $wyz\in
R^-$, while if $\alpha\neq f(y)$, then $\alpha\not\in \{f(x),f(y)\}$, in which
case (S2) fails, since $f(x)\neq f(y)$ and $xyz\in R^+$. This contradicts our
choice of $v$, since $v<z$.
\smallskip

\noindent{\bf Case 2.3:} Suppose that $abv$ was added to $R^+$ using (X5).  Then
there exist vertices $x,y,w,z$ where $\{x,y,w\}<z$ such that either $xyz,wyz\in
R^+$ or $xyz,wyz\in R^-$, and where $\{a,b,v\}=\{x,y,w\}$.  Hence, we deduce
that $f(x),f(y),f(w)$ are pairwise distinct.  Adversary plays from $z$ as
follows: for $z$, if $xyz,wyz\in R^+$, choose any value different from $f(y)$;
if $xyz,wyz\in R^-$, choose any value in $\{f(x),f(y),f(w)\}$. For the rest,
play any choices. Let $\alpha$ denote the value chosen for $z$. We claim that
this instance fails (S2) or (S3) at $z$.  Namely, if $xyz,wyz\in R^+$, then
$\alpha\neq f(y)$ and also $\alpha\neq f(x)$ or $\alpha\neq f(w)$, since
$f(x)\neq f(w)$. Thus (S2) fails, since either $\alpha\not\in\{f(x),f(y)\}$ and
$xyz\in R^+$, or $\alpha\not\in\{f(w),f(y)\}$ and $wyz\in R^+$. Similarly, if
$xyz,wyz\in R^-$, then $\alpha\in\{f(x),f(y),f(w)\}$. Thus if
$\alpha\in\{f(x),f(y)\}$, then (S3) fails, since $xyz\in R^-$, while if
$\alpha=f(w)$, then (S3) fails, since $wyz\in R^-$.  This contradicts the choice
of $v$, since $v<z$.  \smallskip

\noindent{\bf Case 2.4:} Suppose that $abv$ was added to $R^+$ using (X6). Then
there are vertices $x,y,q,w,z$ where $\{x,y,w\}<q<z$ such that either
$xyz,wqz\in R^+$ or $xyz,wqz\in R^-$, and where either $\{a,b,v\}=\{x,y,w\}$ or
$\{a,b,v\}=\{x,y,q\}$. In either case, we have $f(x)\neq f(y)$.  Adversary
plays from $q$ as follows: if $f(w)\in\{f(x),f(y)\}$, then choose
$f(q)$ for $q$; otherwise choose any value different from $f(w)$. Let $\gamma$ denote
the value chosen for $q$. Then for $z$, if $xyz,wqz\in R^-$, choose any value in
$\{f(x),f(y),f(w),\gamma\}$; if $xyz,wqz\in R^+$ and $f(w)\in\{f(x),f(y)\}$,
choose any value different from $f(w)$; otherwise choose any value different
from $\gamma$.  Let $\alpha$ denote the value chosen for $z$.  Note that
$\gamma\neq f(w)$. Indeed, if $f(w)\not\in\{f(x),f(y)\}$, then $\gamma\neq f(w)$
by our choice. If $f(w)\in \{f(x),f(y)\}$, then
$\{a,b,v\}\neq\{x,y,w\}$, since $f(a),f(b),f(v)$ are pairwise distinct.  Thus
$\{a,b,v\}=\{x,y,q\}$ implying that $f(x),f(y),f(q)$ are pairwise distinct; so
$f(w)\neq\gamma=f(q)$, since $f(w)\in \{f(x),f(y)\}$.

We claim that this instance fails (S2) or (S3) at $z$.  Namely, if $xyz,wqz\in
R^-$, then $\alpha\in\{f(x),f(y),f(w),\gamma\}$ and so (S3) either fails because
$\alpha\in\{f(x),f(y)\}$ while $xyz\in R^-$, or it fails because
$\alpha\in\{f(w),\gamma\}$ and $f(w)\neq \gamma$ while $wqz\in R^-$.  If
$xyz,wqz\in R^+$ and $f(w)\in\{f(x),f(y)\}$, then $\alpha\neq f(w)$ and
$\gamma=f(q)\not\in\{f(x),f(y)\}$; thus (S2) fails either because
$\alpha\not\in\{f(x),f(y)\}$ while $xyz\in R^+$, or
$\{\alpha,f(w)\}=\{f(x),f(y)\}$ and $\gamma\not\in\{f(x),f(y)\}$ while $wqz\in
R^+$. Finally, if $xyz,wqz\in R^+$ and $f(w)\not\in\{f(x),f(y)\}$, then
$\alpha\neq\gamma$, and (S2) fails either because $\alpha=f(w)$ while $xyz\in
R^+$, or because $\alpha\neq f(w)$ and $f(w)\neq\gamma\neq\alpha$, while $wqz\in
R^+$.  This contradicts the choice of $v$, since $v<z$.
\medskip

\noindent{\bf Case 3:} Suppose that (S3) fails at $v$. Then there are vertices
$a,b$ with $a<b<v$ such that $abv\in R^-$ while $f(a)\neq f(b)$ and
$f(v)\in\{f(a),f(b)\}$. Since the set $R^-$ is initially empty, the triple $abv$
was added to $R^-$ by one of the rules (X1), (X3), or (X7).  \smallskip

\noindent{\bf Case 3.1:} Suppose that $abv$ was added to $R^-$ using (X1). Then
$av,bv\in F$ and it follows that (S1) fails at $v$. Namely, if $f(v)=f(a)$, then
(S1) fails, since $av\in F$, while if $f(v)=f(b)$, then (S1) fails, since $bv\in
F$. However, we assume that (S1) does not fail at $v$ (as this leads to Case 1),
a contradiction.
\smallskip

\noindent{\bf Case 3.2:} Suppose that $abv$ was added to $R^-$ using (X3).  Then
there exist vertices $x,y,w,z$ where $\{x,y,w\}<z$ and $\{x,y\}<w$ such that
$wz\in F$ and $xyz\in R^+$, and where $\{a,b,v\}=\{x,y,w\}$.  Since $a<b<v$, we
have $\{a,b\}=\{x,y\}$ and $v=w$. In particular, we deduce that
$f(w)\in\{f(x),f(y)\}$.  Adversary plays from $z$ as follows: for $z$, if $f(w)$
offered, choose this value; otherwise, choose any value different from $f(x)$
and $f(y)$.  Let $\alpha$ denote the value chosen for $z$. Note that this choice
is always possible, since Prover offers for $z$ two distinct values; if neither
is $f(w)$, then at least one of them is distinct from both $f(x)$ and $f(y)$,
since $f(w)\in\{f(x),f(y)\}$. For the rest, Adversary play any choices. We claim
that this instance fails (S1) or (S2) at $z$. Namely, if $\alpha=f(w)$, then
(S1) fails at $z$, since $wz\in F$, while if $\alpha\neq f(w)$, then
$\alpha\not\in\{f(x),f(y)\}$, in which case (S2) fails, since $f(x)\neq f(y)$
and $xyz\in R^+$. This contradicts the choice of $v$, since $v<z$.
\smallskip

\noindent{\bf Case 3.3:} Suppose that $abv$ was added to $R^-$ using (X7). Then
there are vertices $x,y,q,w,z$ where $\{x,y,w\}<q<z$ such that either $xyz\in
R^+$ and $wqz\in R^-$, or $xyz\in R^-$ and $wqz\in R^+$, and where either
$\{a,b,v\}=\{x,y,q\}$, or where $\{x,y\}<w$ and $\{a,b,v\}=\{x,y,w\}$.  Since
$a<b<v$, we deduce that $\{a,b\}=\{x,y\}$ and $v\in\{w,q\}$.  In particular,
$f(x)\neq f(y)$ and either $f(w)\in\{f(x),f(y)\}$ or $f(q)\in\{f(x),f(y)\}$.
Adversary plays from $q$ as follows: if $f(w)\not\in\{f(x),f(y)\}$, then choose
$f(q)$ for $q$; otherwise, choose any value different from $f(w)$.  Let $\gamma$
denote the value chosen for $q$.  Then for $z$, if $xyz\in R^+$ and $wqz\in
R^-$, choose $f(w)$ or $\gamma$ if offered, else choose any value distinct from
$f(x)$ and $f(y)$.   Note that this choice is always possible, since in the
latter case Prover offers two distinct values, neither of which is $f(w)$,
$\gamma$, while either $f(w)\in\{f(x),f(y)\}$ or $\gamma=f(q)\in\{f(x),f(y)\}$.
Similarly, if $xyz\in R^-$ and $wqz\in R^+$, we have Adversary choose $f(x)$ or
$f(y)$ if offered, and else choose any value distinct from $f(w)$ and $\gamma$.
Again, this choice is always possible, since $\{f(w),\gamma\}\cap
\{f(x),f(y)\}\neq\emptyset$. Let $\alpha$ denote the value chosen for $z$.  For
the rest, Adversary plays any choices.  Note that $\gamma\neq f(w)$. Indeed, if
$f(w)\in\{f(x),f(y)\}$, then $\gamma\neq f(w)$ by our choice. If
$f(w)\not\in\{f(x),f(y)\}$, then $f(q)\in\{f(x),f(y)\}$ and $\gamma=f(q)$; thus
$\gamma\neq f(w)$, since $\gamma$ is in $\{f(x),f(y)\}$ while $f(w)$ is not.

We claim that this instance fails (S2) or (S3) at $z$. Namely, if $xyz\in R^+$
and $wqz\in R^-$, then either (S3) fails, since $\alpha\in\{f(w),\gamma\}$ while
$wqz\in R^-$, or (S2) fails, since $\alpha\not\in\{f(x),f(y)\}$ while $xyz\in
R^+$. Similarly, if $xyz\in R^-$ and $wqz\in R^+$, then either (S3) fails, since
$\alpha\in\{f(x),f(y)\}$ while $xyz\in R^-$, or (S2) fails, since
$\alpha\not\in\{f(w),\gamma\}$ while $wqz\in R^+$.  This contradicts the choice
of $v$, since $v<z$.  \medskip

This exhausts all possibilities. Therefore no such instance of the game exists,
which proves (ii)$\Rightarrow$(iii).
\medskip

{\bf (iii)$\Rightarrow$(iv)}: Assume that Prover has a strategy as described in
(iii), but (iv) fails, i.e.  there exists a triple $xyz\in R^+$ such that
$x<y<z$ and either $xz\in F$, or $yz\in F$, or $xwz\in R^-$ for some $w<z$, or
$ywz\in R^-$ for some $y<w<z$.  We show that this is impossible. Namely, we show
that there is a way that Adversary can play to violate the conditions of (iii).
As usual, we let $f$ denote the mapping produced during the game.

Suppose first that $xz\in F$ or $yz\in F$.  Adversary plays as follows: until
the game reaches $y$, Adversary plays any choices. When the game reaches $y$,
Prover offers two values for $y$; from the two, Adversary chooses, as the value
$f(y)$, any offered value that is different from $f(x)$. Then Adversary again
plays any choices until the game reaches $z$ when Prover offers two distinct
values for $z$.  If any of the two values is not in $\{f(x),f(y)\}$, then
Adversary chooses this value to be the value $f(z)$. Otherwise, he chooses
$f(x)$ if $xz\in F$, and chooses $f(y)$ if $yz\in F$.  For the rest, Adversary
plays any choices.  It follows that the mapping $f$ fails to satisfy (S1) or
(S2).  Namely, if $f(z)\not\in\{f(x),f(y)\}$, then (S2) fails, since $f(x)\neq
f(y)$ and $xyz\in R^+$. If $f(z)\in\{f(x),f(y)\}$, then either $f(z)=f(x)$ in
case $xz\in F$ and so (S1) fails, or $f(z)=f(y)$ in case $yz\in F$ and so (S1)
fails again.  This contradicts our assumption (iii).

Now, assume that $ywz\in R^-$ for some $y<w<z$. Adversary again chooses $f(x)$
and $f(y)$ to be distinct, and then chooses $f(w)$ to be distinct from $f(y)$.
When $z$ is reached, Adversary chooses $f(y)$ or $f(w)$ if offered by Prover,
and else he chooses any value distinct from $f(x)$.  For the rest, Adversary
plays any choices. It follows that (S2) or (S3) fails for $f$. Namely, if
$f(z)\in\{f(y),f(w)\}$, then (S3) fails, since $f(y)\neq f(w)$ and $ywz\in R^-$.
If $f(z)\not\in\{f(y),f(w)\}$, then $f(z)\neq f(x)$ and (S2) fails,
since $f(x)\neq f(y)$ and $xyz\in R^+$.  This contradicts~(iii).

Lastly, assume that $xwz\in R^-$ where $w<z$ (possibly $w=y$).  Adversary
chooses $f(x)$ and $f(w)$ to be distinct and also chooses $f(y)$ so that $f(x)$ and
$f(y)$ are distinct (possibly $y=w$). When $z$ is reached, Adversary chooses
$f(x)$ or $f(w)$ if offered by Prover, and else he chooses any value distinct
from $f(y)$.  For the rest, Adversary plays any choices. Again, we have that
(S2) or (S3) fails for $f$. Namely, if $f(z)\in \{f(x),f(w)\}$, then (S3) fails,
since $f(x)\neq f(w)$ and $xwz\in R^-$. If $f(z)\not\in \{f(x),f(w)\}$, then
also $f(z)\neq f(y)$ in which case (S2) fails, since $f(x)\neq f(y)$ and
$f(z)\not\in\{f(x),f(y)\}$, while $xyz\in R^+$.  This again contradicts (iii).

This concludes the proof of (iii)$\Rightarrow$(iv).\medskip

{\bf (iv)$\Rightarrow$(iii):}
Assume (iv). We describe a strategy for Prover that will satisfy (iii).
As usual, let $f(\cdot)$ denote the values chosen by Adversary during the game
(partial mapping from $V(G)$ to $\{1,2,3,4\}$).
When asked to offer values for $z$, Prover offers values as follows.
\smallskip

\begin{compactenum}[(1)]
\item If there exist $x,y\in V(G)$ where
\begin{compactitem}
\item $x<y<z$
\item $xyz\in R^+$
\item $f(x)\neq f(y)$
\end{compactitem}\smallskip

then Prover offers $\{f(x),f(y)\}$.\smallskip

\item Else if there exist $x,y\in V(G)$ where 
\begin{compactitem}
\item $x<y<z$
\item $xyz\in R^-$
\item $f(x)\neq f(y)$
\end{compactitem}
\smallskip

then Prover offers $\{1,2,3,4\}\setminus\{f(x),f(y)\}$.
\smallskip

\item Else if there exists $x\in V(G)$ with $x<z$ and $xz\in F$, then Prover
offers any two values different from $f(x)$.
\smallskip

\item Else Prover offers any two values.
\end{compactenum}
\smallskip

We prove that this strategy satisfies the conditions of (iii).  For
contradiction, suppose that Adversary can play against this strategy so that the
resulting mapping $f$ fails one of the conditions (S1)-(S3).

Consider the first point of the game when the value $f(z)$ was assigned to $z$
causing one of (S1)-(S3) to fail. Recall that we assume (iv). We examine the
three possibilities as follows.  \medskip

\noindent{\bf Case 1:}
Suppose that (S1) fails when the value is chosen for $z$. Namely, suppose that
there is $a\in V(G)$ with $a<z$ where $az\in F$ and $f(a)=f(z)$. This
means that $f(a)$ was one of the values offered by Prover for $z$. 
Recall that Prover offered values for $z$ in steps (1)-(4) in that order.
\smallskip

\noindent{\bf Case 1.1:}
Suppose that Prover offered values for $z$ in step (1). Then there exist
vertices $x,y$ where $x<y<z$ such that $xyz\in R^+$ and $f(x)\neq f(y)$, and
Prover offered $\{f(x),f(y)\}$ for $z$. Thus $f(a)\in\{f(x),f(y)\}$, since
$f(a)=f(z)$.  It follows that $a\not\in\{x,y\}$, since otherwise we contradict
(iv).  If $\{x,y\}<a$, then we have $xya\in R^-$ by (X3). But then (S3) is
violated at $a$, since $f(a)\in\{f(x),f(y)\}$. Similarly, if $\{x,y\}\not<a$,
then we have $xa,ya\in F$ by (X3), and (S1) is violated either at $y$ if
$f(a)=f(y)$, or else at $a$ or $x$ if $f(a)=f(x)$. This contradicts our choice
of $z$, since $\{x,y,a\}<z$.  \smallskip

\noindent{\bf Case 1.2:}
Suppose that Prover offered values for $z$ in step (2). Then there exist
vertices $x,y$ where $x<y<z$ such that $xyz\in R^-$ and $f(x)\neq f(y)$, and
Prover offered for $z$ the set $\{1,2,3,4\}\setminus\{f(x),f(y)\}$. Since $f(z)$
was chosen from this set, we have $f(z)\not\in\{f(x),f(y)\}$. Recall that
$f(a)=f(z)$. Thus also $f(a)\not\in\{f(x),f(y)\}$ and hence $a\not\in\{x,y\}$.
This implies by (X2) that $xya\in R^+$. But since $f(a)\not\in\{f(x),f(y)\}$ and
$f(x)\neq f(y)$, we notice that $f(a),f(x),f(y)$ are pairwise distinct, and
hence, (S2) is violated at either $y$ or $a$, since $xya\in R^+$. This
contradicts our choice of $z$, since $\{y,a\}<z$.  \smallskip

\noindent{\bf Case 1.3:}
Suppose that Prover offered values for $z$ in step (3). Then there exists a
vertex $x$ where $x<z$ such that $xz\in F$ and Prover offered for $z$ a set of
two distinct values, neither of which was $f(x)$. Since $f(z)$ was chosen from
this set, we have $f(z)\neq f(x)$. Recall that $f(a)=f(z)$. Thus $f(a)\neq f(x)$
and $a\neq x$.  From this we deduce using (X1) that $xaz\in R^-$. Consequently,
Prover should have offered values in step (2), never reaching step (3), since
$f(a)\neq f(x)$ and $xaz\in R^-$. Thus Prover never reached step (3), a
contradiction.  \smallskip

\noindent{\bf Case 1.4:}
Suppose that Prover offered values for $z$ in step (4). Since step (4) was
reached, there is no $x$ such that $x<z$ and $xz\in F$. Thus it is impossible
that (S1) failed when the value for $z$ was chosen, a contradiction.  \medskip

\noindent{\bf Case 2:}
Suppose that (S2) fails when the value is chosen for $z$.  Namely, suppose that
there exist vertices $a,b$ where $a<b<z$ and $abz\in R^+$ such that $f(a)\neq
f(b)$ and $f(z)\not\in\{f(a),f(b)\}$.  Note that this implies that Prover
offered values for $z$ in step (1), since we may always take $x=a$
and $y=b$ to satisfy the conditions of step (1). Thus we only need to consider
this possibility.  Namely, we have that there exist vertices $x,y$ where $x<y<z$
such that $xyz\in R^+$ and $f(x)\neq f(y)$, and Prover offered $\{f(x),f(y)\}$
for $z$. Thus $f(z)\in\{f(x),f(y)\}$. Recall that $f(z)\not\in\{f(a),f(b)\}$.
Hence $\{a,b\}\neq\{x,y\}$. Moreover, since $f(a)\neq f(b)$ and $f(x)\neq f(y)$,
it follows that $\{f(a),f(b)\}\neq \{f(x),f(y)\}$, since $\{f(x),f(y)\}$
contains $f(z)$ while $\{f(a),f(b)\}$ does not.

Assume first that $\{x,y\}$ is disjoint from $\{a,b\}$.  If $y<b$, then we
deduce using (X6) that $xya,xyb\in R^+$. This means that (S2) fails at $b$ or at
one of $y,a$, since $f(x)\neq f(y)$ and $\{f(a)$, $f(b)\}\neq\{f(x),f(y)\}$.
Similarly, if $b<y$, we deduce using (X6) that $abx,aby\in R^+$, and (S2) fails
at $y$ or at one of $b,x$, since $f(a)\neq f(b)$ and $\{f(x), f(y)\} \neq
\{f(a), f(b)\}$.  This contradicts our choice of $z$, since $\{x,y,a,b\}<z$.

So we may assume that $\{x,y\}$ intersects $\{a,b\}$.  If $y\in\{a,b\}$, then
$xab\in R^+$ by (X5). Recall that $\{f(a),f(b)\}\neq\{f(x),f(y)\}$. Since
$y\in\{a,b\}$, we deduce that $f(x)\not\in\{f(a),f(b)\}$. This implies that
$f(x),f(a),f(b)$ are pairwise distinct, since also $f(a)\neq f(b)$. Thus (S2)
fails at $b$, since $xab\in R^+$.  Similarly, if $x\in\{a,b\}$, then $aby\in
R^+$ by (X5) and we have $f(y)\not\in\{f(a),f(b)\}$.  Hence, $f(a),f(b),f(y)$
are pairwise distinct and so (S2) fails at $y$ or $b$, since $aby\in R^+$. This
contradicts our choice of $z$, since $\{y,b\}<z$.
\medskip

\noindent{\bf Case 3:}
Suppose that (S3) fails when the value is chosen for $z$. Namely, suppose that
there exist vertices $a,b$ where $a<b<z$ and $abz\in R^-$ such that $f(a)\neq
f(b)$ and $f(z)\in\{f(a),f(b)\}$. Note that this implies that Prover offered
values for $z$ in either step (1) or step (2), since we may always take $x=a$
and $y=b$ to satisfy the conditions of step (2). Thus we only need to consider
the steps (1) and (2).  \smallskip

\noindent{\bf Case 3.1:}
Suppose that Prover offered values for $z$ in step (1). Then there exist
vertices $x,y$ where $x<y<z$ such that $xyz\in R^+$ and $f(x)\neq f(y)$, and
Prover offered $\{f(x),f(y)\}$ for $z$. Thus $f(z)\in \{f(x),f(y)\}$.
Recall that $f(z)\in\{f(a),f(b)\}$. We deduce that
$\{f(x),f(y)\}\cap\{f(a),f(b)\}\neq\emptyset$.

Assume first that $\{a,b\}$ and $\{x,y\}$ are disjoint.  If $y<b$, then we
deduce using (X7) that $xyb\in R^-$ and either $xya\in R^-$ if $\{x,y\}<a$, or
else $xa,ya\in F$.  Thus if $f(b)\in\{f(x),f(y)\}$, then (S3) fails at $b$,
since $f(x)\neq f(y)$ and $xyb\in R^-$. So we may assume that
$f(b)\not\in\{f(x),f(y)\}$ which yields that $f(a)\in\{f(x),f(y)\}$, since
$\{f(a),f(b)\}\cap \{f(x),f(y)\}\neq\emptyset$. Thus if $\{x,y\}<a$, we have
$xya\in R^-$ and so (S3) fails at $a$, since $f(x)\neq f(y)$. If $\{x,y\}\not<
a$, we have $xa,ya\in F$ in which case (S1) fails at either $y$ or one of $x,a$.
Similarly, if $b<y$, we have by (X7) that $aby\in R^-$ and either $abx\in R^-$
if $\{a,b\}<x$, or $ax,bx\in F$ if otherwise.  Thus either (S3) fails at $y$ if
$f(y)\in\{f(a),f(b)\}$, or we have $f(x)\in\{f(a),f(b)\}$ in which case either
(S3) fails at $x$ if $\{a,b\}<x$, or (S1) fails at $a$ or $b$ or $x$ if
$\{a,b\}\not< x$.  This contradicts our choice of $z$, since $\{x,y,a,b\}<z$.

Thus we may assume that $\{a,b\}$ intersects $\{x,y\}$. Recall that $a<b$ and
$x<y$. We observe that if \mbox{$x\in\{a,b\}$} or $y=a$, then we contradict
(iv), the second or third condition thereof, respectively. Thus it follows that
\mbox{$x\not\in\{a,b\}$} and $y=b$.  From this we deduce using (X4) that $ax\in
F$ and $axy\in R^+$.  We recall that $f(x)\neq f(y)$ and $f(a)\neq f(b)$. Since
$y=b$, we deduce that $f(y)\not\in\{f(a),f(x)\}$. Thus either $f(a)=f(x)$ and
(S1) fails at one of $a$, $x$, since $ax\in F$, or we have $f(a)\neq f(x)$ in
which case (S2) fails at $y$, since $axy\in R^+$ and $f(y)\not\in\{f(a),f(x)\}$.
This again contradicts our choice of $z$, since $\{a,x,y\}<z$.  \smallskip

\noindent{\bf Case 3.2:}
Suppose that Prover offered values for $z$ in step (2). Then there exist
vertices $x,y$ where $x<y<z$ such that $xyz\in R^-$ and $f(x)\neq f(y)$, and
Prover offered for $z$ the set $\{1,2,3,4\}\setminus\{f(x),f(y)\}$. Since $f(z)$
was chosen from this set, we have $f(z)\not\in\{f(x),f(y)\}$.
Recall that $f(z)\in\{f(a),f(b)\}$ and $f(a)\neq f(b)$. We deduce that
$\{f(a),f(b)\}\neq \{f(x),f(y)\}$ and so $\{a,b\}\neq\{x,y\}$.
Now we proceed exactly as in Case 2.

If $\{x,y\}$ is disjoint from $\{a,b\}$, we consider two cases: $y<b$ or $b<y$.
If $y<b$, then $xya,xyb\in R^+$ by (X6), while if $b<y$, we have $abx,aby\in
R^+$. In either case, we deduce that (S2) fails at one of $x$, $y$, $a$, $b$,
since $\{f(a),f(b)\}\neq\{f(x),f(y)\}$.  If $\{x,y\}\cap \{a,b\}\neq\emptyset$,
then we again have two cases: $y\in \{a,b\}$ or $x\in\{a,b\}$. If $y\in\{a,b\}$,
we have $xab\in R^+$ by (X5) and we deduce that $f(x)$, $f(a)$, $f(b)$ are
pairwise distinct. Thus (S2) fails at $b$, since $xab\in R^+$.  If $x\in
\{a,b\}$, then $aby\in R^+$ by (X5) and $f(y)$, $f(a)$, $f(b)$ are pairwise
distinct. Thus (S2) fails at $b$ or $y$, since $aby\in R^+$. This contradicts
our choice of $z$, since $\{y,b\}<z$.
\medskip

This exhausts all possibilities. Thus we conclude that no such vertex $z$ exists
which proves that the strategy for Prover described in steps (1)-(4) is indeed a
strategy satisfying the conditions of (iii). Therefore (iv)$\Rightarrow$(iii).

This completes the proof of Theorem \ref{thm:2csp-k4-char}.\hfill\qed\bigskip

With this characterization, we can now prove Theorem \ref{thm:2csp-k4} as
follows.\medskip

\subsection{Proof of Theorem \ref{thm:2csp-k4}}
By Theorem \ref{thm:2csp-k4-char}, it suffices to construct the sets $F$, $R^+$,
and $R^-$, and check the conditions of item (iv) of the said theorem.  This can
clearly be accomplished in polynomial time, since each of the three sets
contains at most $n^3$ elements, where $n$ is the number of variables in the
input formula, and elements are only added (never removed) from the sets.  Thus
either a new pair (triple) needs to be added as follows from one of the rules
(X1)-(X7), or we can stop and the output the resulting sets. 
\hfill\qed

\section{Hardness of $\{n\}$-CSP($\bK_{2n}$) for $n\geq 3$}\label{sec:ncsp-2n}

\begin{theorem}\label{thm:ncsp-2n}
$\{n\}$-CSP($\mathbb{K}_{2n}$) is Pspace-complete for all $n\geq 3$.
\end{theorem}

\begin{figure}[t]
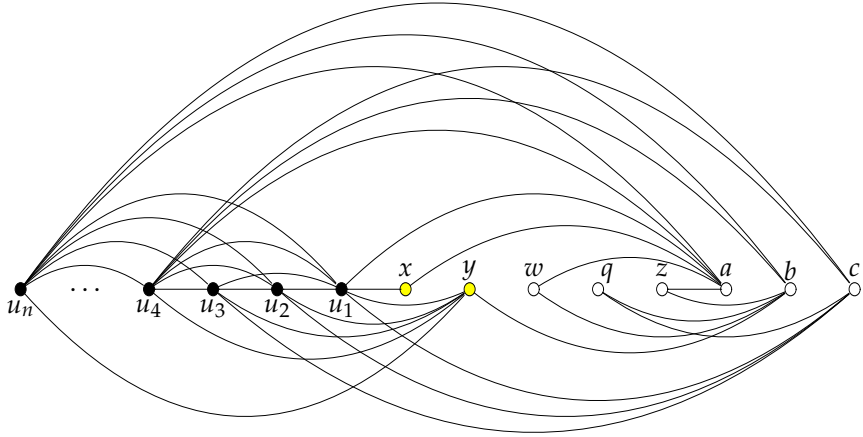
\centering
$\xy/r2pc/:
(-3,0)*[o][F**:black]{\phantom{s}}="2";
(-2,0)*{\ldots};
(-1,0)*[o][F**:black]{\phantom{s}}="3";
(0,0)*[o][F**:black]{\phantom{s}}="4";
(1,0)*[o][F**:black]{\phantom{s}}="5";
(2,0)*[o][F**:black]{\phantom{s}}="6";
(3,0)*[o][F**:yellow]{\phantom{s}}="x";
(4,0)*[o][F**:yellow]{\phantom{s}}="y";
(5,0)*[o][F]{\phantom{s}}="w";
(6,0)*[o][F]{\phantom{s}}="q";
(7,0)*[o][F]{\phantom{s}}="z";
(8,0)*[o][F]{\phantom{s}}="a";
(9,0)*[o][F]{\phantom{s}}="b";
(10,0)*[o][F]{\phantom{s}}="c";
{\ar@{-} "4";"5"};
{\ar@{-}@/^0.5pc/ "4";"6"};
{\ar@{-} "5";"6"};
{\ar@{-}@/^3pc/ "6";"a"};
{\ar@{-}@/_4.5pc/ "4";"c"};
{\ar@{-}@/_4pc/ "5";"c"};
{\ar@{-}@/_3.5pc/ "6";"c"};
{\ar@{-} "6";"x"};
{\ar@{-}@/_0.5pc/ "6";"y"};
{\ar@{-}@/^2pc/ "x";"a"};
{\ar@{-}@/_2pc/ "y";"b"};
{\ar@{-}@/^1pc/ "w";"a"};
{\ar@{-}@/_1.5pc/ "w";"b"};
{\ar@{-}@/_1pc/ "q";"b"};
{\ar@{-}@/_1.5pc/ "q";"c"};
{\ar@{-} "z";"a"};
{\ar@{-}@/_0.5pc/ "z";"b"};
"x"+(0,0.3)*{x};
"y"+(0,0.3)*{y};
"w"+(0,0.3)*{w};
"q"+(0.15,0.25)*{q};
"z"+(0,0.3)*{z};
"a"+(0,0.3)*{a};
"b"+(0,0.3)*{b};
"c"+(0,0.3)*{c};
"4"+(0,-0.3)*{u_3};
"5"+(0,-0.3)*{u_2};
"6"+(0,-0.3)*{u_1};
"2"+(0,-0.3)*{u_n};
"3"+(0,-0.3)*{u_4};
{\ar@{-}@/^0.75pc/ "2";"3"};
{\ar@{-}@/^1.5pc/ "2";"4"};
{\ar@{-}@/^2.25pc/ "2";"5"};
{\ar@{-}@/^3pc/ "2";"6"};
{\ar@{-} "3";"4"};
{\ar@{-}@/^0.75pc/ "3";"5"};
{\ar@{-}@/^1.5pc/ "3";"6"};
{\ar@{-}@/^7pc/ "2";"a"};
{\ar@{-}@/^8pc/ "2";"b"};
{\ar@{-}@/^9pc/ "2";"c"};
{\ar@{-}@/^5pc/ "3";"a"};
{\ar@{-}@/^6pc/ "3";"b"};
{\ar@{-}@/^7pc/ "3";"c"};
{\ar@{-}@/_1.5pc/ "4";"y"};
{\ar@{-}@/_1.1pc/ "5";"y"};
{\ar@{-}@/_4pc/ "2";"y"};
{\ar@{-}@/_2.2pc/ "3";"y"};
\endxy$
\caption{The edge gadget (here, as an example, $x$ is an $\exists$ vertex while
$y$ is a $\forall$
vertex).\label{fig:ncsp-2n}}
\end{figure}

The template $\mathbb{K}_{2n}$ consists of vertices $\{1,2,\ldots,2n\}$ and all
possible edges between distinct vertices. We shall call these vertices {\em
colours}.
We describe a reduction from QCSP($\mathbb{K}_n$)=$\{1,n\}$-CSP($\mathbb{K}_n$)
to $\{n\}$-CSP($\mathbb{K}_{2n}$).
Consider an instance of QCSP($\mathbb{K}_n$), namely a formula $\Psi$
where\smallskip

\centereq{
$\Psi=\exists^{\geq b_1}\,v_1\,\exists^{\geq b_2}\,v_2\,\ldots\,\exists^{\geq b_N}\,v_N\,\psi$
}\smallskip

\noindent where each $b_i\in\{1,n\}$. As usual (see Definition
\ref{def:instance}), let $G$ denote the graph $\D_\psi$ with vertex set
$\{v_1,\ldots,v_N\}$ and edge set $\{v_iv_j~|~E(v_i,v_j)$ appears in $\psi\}$.

We construct an instance $\Phi$ of $\{n\}$-CSP($\mathbb{K}_{2n}$) with the property
that $\Psi$ is a yes-instance of QCSP($\mathbb{K}_n$) if and only if $\Phi$ is a
yes-instance of $\{n\}$-CSP($\mathbb{K}_{2n}$).  

In short, we shall model the $n$-colouring using $2n-1$ colours, $n-1$ of which will
treated as {\em don't care} colours (vertices coloured using any of such colours
will be ignored). We make sure that the colourings where no vertex is assigned a
don't-care colour precisely model all colourings that we need to check to verify
that $\Psi$ is a yes-instance.

We describe $\Phi$ by giving a graph $H$ together with a total order of its
vertices with the usual interpretation that the vertices are the variables of
$\Phi$, the total order is the order of quantification of the variables, and the
edges of $H$ define the conjunction of predicates $E(\cdot,\cdot)$ which forms
the quantifier-free part $\phi$ of $\Phi$.

We start constructing $H$ by adding the vertices $v_1,v_2,\ldots,v_N$ and no edges.
Then we add new vertices $u_1,u_2,\ldots,u_n$  and make them pairwise adjacent.

We make each $v_i$ adjacent to $u_1$, and if $b_i=n$ (i.e. if $v_i$
was quantified $\forall$), then we also make $v_i$ adjacent to
$u_2,u_3,\ldots,u_n$.  

We complete $H$ by introducing for each edge $xy\in E(G)$, a gadget
consisting of new vertices $w,q,z,a,b,c$ with edges $wa,wb,qb,qc,za,zb$, and we
connect this gadget to the rest of the graph as follows: we make $x$ adjacent to
$a$, make $y$ adjacent to $b$, make $a$ adjacent to $u_1$, make $c$ adjacent to
$u_1,u_2,u_3$, and make each of $a,b,c$ adjacent to $u_4,\ldots,u_n$.  We refer
to Figure \ref{fig:ncsp-2n} for an illustration.

The total order of $V(H)$ first lists $u_1,u_2,\ldots, u_n$, then
$v_1,v_2,\ldots,v_N$ (exactly in the same order as quantified in $\Psi$), and
then lists the remaining vertices of each gadget, in turn, as depicted in Figure
\ref{fig:ncsp-2n} (listing $w,q,z,a,b,c$ in this order).

We consider the game $\mathscr{G}(\Phi,\mathbb{K}_{2n})$ of Prover and Adversary
played on $\Phi$ where Prover and Adversary take turns, for each variable in
$\Phi$ in the order of quantification, respectively providing a set of $n$ colours
and choosing a colour from the set. Prover wins if this process leads to a
proper $2n$-colouring of $H$ (no adjacent vertices receive the same colour),
otherwise Prover loses and Adversary wins. The formula $\Phi$ is a
yes-instance if and only if Prover has a winning strategy.

Without loss of generality (up to renaming colours), we may assume that the
vertices $u_1,u_2,\ldots,u_n$ get assigned colours $n+1,n+2,\ldots,2n$,
respectively, i.e. each $u_i$ gets colour $n+i$. (The edges between these
vertices make sure that Prover must offer distinct colours while Adversary has
no way of forcing a conflict, since there are $2n$ colours available.)

The claim of Theorem \ref{thm:ncsp-2n} will then follow from the following two
lemmas.  (details omitted due to space restrictions -- see the appendix)

\begin{lemma}\label{clm:ncsp-1}
If Adversary is allowed to choose for the vertices $x,y$ in the edge gadget
(Figure \ref{fig:ncsp-2n}) the same colour from $\{1,2,\ldots,n\}$, then
Adversary wins.  If Adversary is allowed to choose $n+1$ for $x$ or $y$, then
Adversary also wins.

In all other cases, Prover wins.
\end{lemma}

\begin{proof}
If Prover offers $n+1$ for $x$ or $y$, then Adversary can choose this colour
(for $x$ or $y$) and Prover immediately loses, since both $x$ and $y$ are
adjacent to $u_1$ which is assumed to be assigned the colour $n+1$. (Prover
loses since the colouring is not proper.)

Assume that $x$ and $y$ are assigned the same colour $i$ from
$\{1,2,\ldots,n\}$.  We describe a winning strategy for Adversary.  Consider the
set of $n$ colours Prover offers for $w$. Since the colours are distinct and
there is $n$ of them, at least one of the colours, denote it $k$, is different
from $i$ and each of $n+1,n+4,n+5,\ldots,2n$.  Adversary chooses the
colour $k$~for~$w$.

Then consider the $n$ colours Prover offers for $q$. If any of the $n$ colours,
denote it $j$, is from $\{1,2,\ldots,n\}$, then Adversary chooses the colour $j$
for $q$, which makes Prover lose when considering the vertex $c$ where Prover
must offer $n$ values different from $n+1,n+2,\ldots,2n$ and from
$j\in\{1,2,\ldots,n\}$, which is impossible.  (Note that $c$ is adjacent to $q$
as well as $u_1,u_2,\ldots,u_n$.)

Therefore we may assume that Prover offers the set $\{n+1,n+2,\ldots,2n\}$
for~$q$.  Let $\ell$ be any colour in the set $\{n+2,n+3\}\setminus\{k\}$.  By
definition, $\ell$ is different from $k$, and clearly also different from $i$,
since $i\in\{1,2,\ldots,n\}$ while $\ell\in\{n+2,n+3\}$.  We make Adversary
choose the colour $\ell$ for $q$.

Now if Prover offers for $z$ a colour, denote it $r$, different from $i$, $k$
and each of $n+1$, $n+4$, $n+5$, \ldots, $2n$, then Adversary chooses this
colour and Prover loses at $a$ when she has to provide $n$ colours 
distinct from $i,k,r,n+1,n+4,n+5,\ldots,2n$, which is impossible.  Similarly,
Prover loses, this time at $b$, if she offers for $z$ a colour different from
$i$, $k$, $\ell$, $n+4$, $n+5$, \ldots, $2n$.  Notice that there is no set of
$n$ values that excludes both these situations, since $i$, $k$, $\ell$, $n+1$,
$n+4$, \ldots, $2n$ are distinct values. This shows that Adversary wins no
matter what Prover does.
\medskip

Now, for the second part of the claim, assume that $x$ and $y$ are either given
distinct colours different from $n+1$, or same colours from
$\{n+2,n+3,\ldots,2n\}$. This time Prover wins no matter what Adversary does.
\medskip

First, assume that $x$ and $y$ have distinct colours $i$ and $j$, respectively
where $i,j\neq n+1$. We consider three cases.
\medskip

\noindent{\bf Case 1}: assume that $j\in\{n+4,n+5,\ldots,2n\}$.  Then we have
Prover offer for the vertices $w$ and $q$ the set
$\{n+1,n+2,\ldots,2n\}$.  Let $k$ be the colour chosen by Adversary for $w$, and
let $\ell$ be the colour chosen for $q$.

If $\{i,k,n+1,n+4,n+5,\ldots,2n\}$
contains $n$ distinct elements, we have Prover offer this set for $z$; otherwise
Prover offers any set of $n$ distinct elements for $z$.  Let $r$ be the colour
chosen for $z$.  Now Prover offers for $a$ any set of $n$ colours disjoint from
$\{i,k,r,n+1,n+4,\ldots,2n\}$, which is possible since
$r\in\{i,k,n+1,n+4,\ldots,2n\}$. For $b$ Prover offers any set of $n$ colours
disjoint from $\{j,k,\ell,r,n+4,\ldots,2n\}$, which is again possible because
$j\in\{n+4,\ldots,2n\}$. Finally, Prover offers $\{1,2,\ldots,n\}$ for $c$.  It
is now easy to see that any choice of Adversary yields a proper colouring and so
Prover wins, as claimed.
\medskip

\noindent{\bf Case 2}: assume that $i\in\{n+4,n+5,\ldots,2n\}$. We similarly
have Prover offer $\{n+1,n+2,\ldots,2n\}$ for both $w$ and $q$, and let $k$ and
$\ell$ be the colours chosen by Adversary for the two vertices. If
$\{j,k,\ell,n+4,n+5,\ldots,2n\}$ contains $n$ distinct elements, Prover offers
this set for $z$; otherwise Prover offers any set of $n$ distinct elements. Just
like in Case 1.1, this now allows us to choose $n$ distinct colours for each of
$a,b,c$ so that none of the colours appears on their neighbours. So again, for
any choice of Adversary, Prover wins as required.  \medskip

\noindent{\bf Case 3}: assume that $i,j\not\in\{n+4,n+5,\ldots,2n\}$.
Recall that $i,j\neq i+1$ and $i\neq j$. Thus we have Prover offer for $w$
the set $\{n+1,i,j,n+4,\ldots,2n\}$ and for $q$ the set $\{n+1,\ldots,2n\}$. Let
$k$ be the colour chosen by Adversary for $w$, and let $\ell$ be the colour
chosen for $q$.

Suppose first that $k\in\{i,n+1,n+4,\ldots,2n\}$.  If
$\{j,k,\ell,n+4,\ldots,2n\}$ contains $n$ distinct elements, Prover offers this
set for $z$; otherwise, she offers any set for~$z$. Again, for each of $a,b,c$,
there are at most $n$ colours used on their neighbours and so Prover can offer
each of $a,b,c$ a set of $n$ colours distinct from their neighbours to get a
proper colouring for any choice of Adversary.

So we may assume that $k=j$. In this case, we have Prover offer the set
$\{i,n+1,n+4,\ldots,2n\}$ for $z$. Again, for each of $a,b,c$ we have $n$
colours distinct from their neighbours and we can thus complete a proper
colouring regardless of Adversary's choices. Thus Prover wins in any situation.
\medskip

That exhausts all possibilities for when $x,y$ have distinct colours different
from $n+1$.  To finish the proof, it remains to consider the case when $x,y$
have the same colour $i$ from $\{n+2,n+3,\ldots,2n\}$ In this case, Prover
offers for the vertices $w$, $q$, $z$ the set $\{n+1,n+2,\ldots,2n\}$, while for
$a$, $b$, $c$ Prover offers the set $\{1,2,\ldots,n\}$.  It is easy to see that
any choice of Adversary yields a proper colouring.  Thus Prover wins as
required.

That concludes the proof.
\end{proof}

\begin{lemma}\label{clm:ncsp-2}
$\Phi$ is a yes-instance of $\{n\}$-CSP($\mathbb{K}_{2n}$) if and only if $\Psi$ is
a yes-instance of QCSP($\mathbb{K}_n$).
\end{lemma}

\begin{proof}
We treat the colours $n+2,n+3,\ldots,2n$ as {\em don't care} colours, while
$1,2,\ldots,n$ will be the actual colours used for colouring $G$.  By Lemma
\ref{clm:ncsp-1}, the edge gadget makes sure that vertices $x,y$ do not receive
the same colours unless at least one of the colours is from
$\{n+2,n+3,\ldots,2n\}$ (the don't-care colours).  This implies that $\Phi$
correctly simulates $\Psi$ whereby Prover offers $\{1,2,\ldots,n\}$ for each
$\forall$ variable of $\Psi$, and offers $\{i,n+2,n+3,\ldots,2n\}$ for each
$\exists$ variable of $\Psi$ where $i\in\{1,2,\ldots,n\}$.  Note that the
construction forces Prover to offer $\{1,2,\ldots,n\}$ for each $\forall$
variable, while for each $\exists$ variable Prover must offer $n$ values
excluding the value $n+1$. In the latter case we may assume that the set of
offered values is of the form $\{i,n+2,n+3,\ldots,2n\}$, where
$i\in\{1,2,\ldots,n\}$, since offering more values from $\{1,2,\ldots,n\}$ makes
it even easier for Adversary to win (has more choices to force a monochromatic
edge).

Thus this shows that $\Phi$ indeed correctly simulates $\Psi$ as required.
\end{proof}

We finish the proof by remarking that the construction of $\Phi$ is polynomial
in the size of $\Psi$. Thus, since QCSP($\mathbb{K}_n)$ is Pspace-hard, so is
$\{n\}$-CSP($\mathbb{K}_{2n}$). 

This completes the proof of Theorem \ref{thm:ncsp-2n}.

\section{Algorithm for $\{1,2\}$-CSP($\bP_\infty$)}\label{sec:12csp-path}

We consider the infinite path $\P_{\infty}$ to be the graph whose vertex set is
$\mathbb{Z}$ and edges are $\{ij~:~|i-j|=1\}$.  In this section, we prove the
following theorem.

\begin{theorem}\label{thm:inf-path}
$\{1,2\}$-CSP($\bP_\infty$) is decidable in polynomial time.
\end{theorem}

An instance to
$\{1,2\}$-CSP($\P_\infty$) is a graph $G=\D_\psi$, a total order $\prec$ on
$V(G)$, and a function $\beta:V(G)\rightarrow\{1,2\}$ where\smallskip

\centereq{$\Psi:=\exists^{\geq\beta(v_1)}\,v_1~\exists^{\geq\beta(v_2)}\,v_2~\cdots~\exists^{\geq
\beta(v_n)}\,v_n~\displaystyle\bigwedge_{v_iv_j\in E(G)} E(v_i,v_j)$}

\subsection{Definitions}

We write $X\prec Y$ if $x\prec y$ for each $x\in X$ and each $y\in Y$. Also, we
write $x\prec Y$ in place of $\{x\}\prec Y$.  A {\em walk} of $G$ is a sequence
$x_1,x_2,\ldots,x_r$ of vertices of $G$ where $x_ix_{i+1}\in E(G)$ for all
$i\in\{1,\ldots,r-1\}$.  A walk $x_1,\ldots,x_r$ is a {\em closed walk} if
$x_1=x_r$.  Write $|Q|$ to denote the {\em length} of the walk $Q$ (number of
edges on $Q$).

\begin{definition}
If $Q=x_1,\ldots,x_r$ is a walk of $G$, we define $\lambda(Q)$ as follows:

\centereq{$\lambda(Q)=|Q|-2\Sum_{i=2}^{r-1}\big(\beta(x_i)-1\big)$}

\end{definition}

\begin{definition}
A walk $x_1,\ldots,x_r$ of $G$ is a \uline{\em looping walk} if
$x_1\neq x_r$ and if $r\geq 3$
\begin{compactenum}[(i)]
\item $\{x_1,x_r\}\prec\{x_2,\ldots,x_{r-1}\}$, and
\item there is $\ell\not\in\{1,r\}$ such that both
$x_1,\ldots,x_\ell$ and $x_\ell,\ldots,x_r$ are looping walks.
\end{compactenum}
\end{definition}

The above is a recursive definition. (The base case is when $r=2$ and $x_1,x_2$
are two distinct adjacent vertices.) Observe that endpoints of a looping walk
are always distinct and never appear in the interior of the walk. Other
vertices, however, may appear on the walk multiple times as long as the walk
obeys (ii). Notably, it is possible that the same vertex is one of
$x_2,\ldots,x_{\ell-1}$ as well as one of $x_{\ell-1},\ldots,x_{r-1}$ where
$\ell$ is as defined in (ii). See Figure \ref{fig:12-csp-path} for examples.

Using looping walks, we define a notion of ``distance'' in $G$ that will
guide Prover in the game.

\begin{definition}
For vertices $u,v\in V(G)$, define $\delta(u,v)$ to be the following:

\centereq{
$\delta(u,v)=\min\Big\{\lambda(Q)~\Big|~$\parbox{9.7cm}{$Q=x_1,\ldots,x_r$ is a
looping walk of $G$ where $x_1=u$ and $x_r=v$}$\Big\}.$}

\noindent If no looping walk between $u$ and $v$ exists, define $\delta(u,v)=\infty$.
\end{definition}

Namely, $\delta(u,v)$ denotes the smallest $\lambda$-value of a looping
walk between $u$ and $v$.  Note that $\delta(u,v)=\delta(v,u)$, since the
definition of a looping walk does not prescribe the order of the endpoints of
the walk.
\medskip

The main structural obstruction in our characterization is the following.

\begin{definition}
A \uline{\em bad walk} of $G$ is a looping walk $Q=x_1,\ldots,x_r$ of $G$ such
that $x_1\prec x_r$ and $\lambda(Q)\leq\beta(x_r)-2$.
\end{definition}

\begin{figure}[h!t]
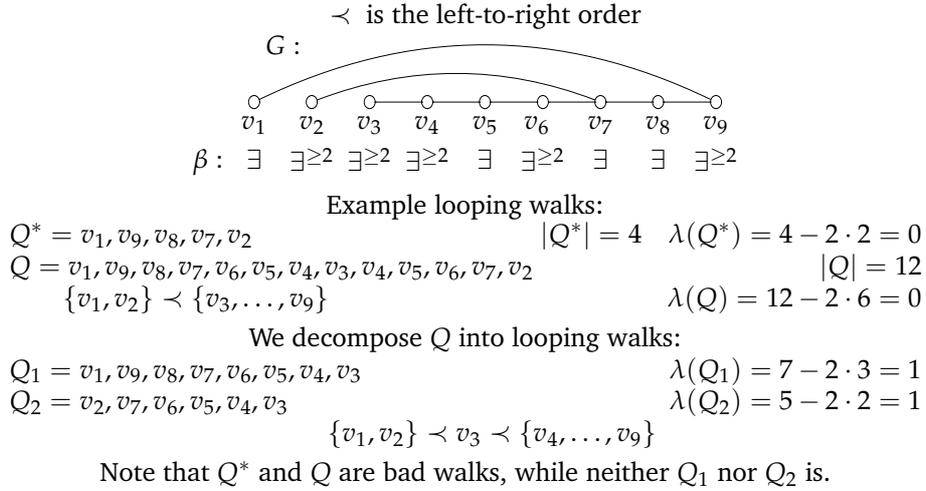

\centering
$\xy/r1.8pc/:
(0.5,1)*{G:};
(4,1.5)*{\prec\mbox{ is the left-to-right order}};
(0,0)*[o][F]{\phantom{s}}="v1";
(1,0)*[o][F]{\phantom{s}}="v2";
(2,0)*[o][F]{\phantom{s}}="v3";
(3,0)*[o][F]{\phantom{s}}="v4";
(4,0)*[o][F]{\phantom{s}}="v5";
(5,0)*[o][F]{\phantom{s}}="v6";
(6,0)*[o][F]{\phantom{s}}="v7";
(7,0)*[o][F]{\phantom{s}}="v8";
(8,0)*[o][F]{\phantom{s}}="v9";
{\ar@{-} "v3";"v4"};
{\ar@{-} "v4";"v5"};
{\ar@{-} "v5";"v6"};
{\ar@{-} "v6";"v7"};
{\ar@{-} "v7";"v8"};
{\ar@{-} "v8";"v9"};
{\ar@{-}@/^1.8pc/ "v1";"v9"};
{\ar@{-}@/^0.9pc/ "v2";"v7"};
"v1"+(0,-0.4)*{v_1};
"v2"+(0,-0.4)*{v_2};
"v3"+(0,-0.4)*{v_3};
"v4"+(0,-0.4)*{v_4};
"v5"+(0,-0.4)*{v_5};
"v6"+(-0.1,-0.4)*{v_6};
"v7"+(0,-0.4)*{v_7};
"v8"+(0,-0.4)*{v_8};
"v9"+(0,-0.4)*{v_9};
(-0.8,-1.05)*{\beta:};
"v1"+(0,-1.0)*{\exists};
"v2"+(0,-1.0)*{\exists^{\geq 2}};
"v3"+(0,-1.0)*{\exists^{\geq 2}};
"v4"+(0,-1.0)*{\exists^{\geq 2}};
"v5"+(0,-1.0)*{\exists};
"v6"+(0,-1.0)*{\exists^{\geq 2}};
"v7"+(0,-1.0)*{\exists};
"v8"+(0,-1.0)*{\exists};
"v9"+(0,-1.0)*{\exists^{\geq 2}};
\endxy$\medskip

Example looping walks:

\parbox{12cm}{\centering
$Q^*=v_1,v_9,v_8,v_7,v_2$\hfill $|Q^*|=4$\quad $\lambda(Q^*)=4-2\cdot 2=0$

$Q=v_1,v_9,v_8,v_7,v_6,v_5,v_4,v_3,v_4,v_5,v_6,v_7,v_2$\hfill $|Q|=12$

\qquad $\{v_1,v_2\}\prec\{v_3,\ldots,v_9\}$\hfill $\lambda(Q)=12-2\cdot 6=0$\smallskip

We decompose $Q$ into looping walks:

$Q_1=v_1,v_9,v_8,v_7,v_6,v_5,v_4,v_3$\hfill$\lambda(Q_1)=7-2\cdot 3=1$

$Q_2=v_2,v_7,v_6,v_5,v_4,v_3$\hfill$\lambda(Q_2)=5-2\cdot 2=1$

\qquad$\{v_1,v_2\}\prec v_3\prec\{v_4,\ldots,v_9\}$\smallskip

 Note that $Q^*$ and $Q$ are bad walks, while neither
$Q_1$ nor $Q_2$ is.
}
\caption{Examples of looping walks.\label{fig:12-csp-path}}
\end{figure}

\subsection{Characterization}

\begin{theorem}\label{thm:12csp-path}
Suppose that $G$ is a bipartite graph. Then the following statements are equivalent.
\begin{compactenum}[(I)]
\item $\P_\infty\models \Psi$
\item Prover has a winning strategy in $\G(\Psi,\P_\infty)$.
\item Prover can play $\G(\Psi,\P_\infty)$ so that in every instance of the
game, the resulting mapping $f$ satisfies the following for all
$u,v\in V(G)$ with $\delta(u,v)<\infty$:\medskip

\mbox{}\hfill$|f(u)-f(v)|\leq \delta(u,v)$\,,\hfill$(\star)$\medskip

\mbox{}\hfill$f(u)+f(v)+\delta(u,v)$ is an even number\,.\hfill$(\triangle)$\medskip

\item There are no $u,v\in V(G)$ where $u\prec v$ such that $\delta(u,v)\leq \beta(v)-2$\,.
\item There is no bad walk in $G$.
\end{compactenum}
\end{theorem}

\subsection{Proof of Theorem \ref{thm:12csp-path}}

We prove the claim by considering individual implications.  The
equivalence (I)$\Leftrightarrow$(II) is proved as Lemma~\ref{lem:game}.  The
equivalence (IV)$\Leftrightarrow$(V) follows immediately from the definitions of
$\delta(\cdot,\cdot)$ and bad walk.  The other implications are proved as
follows. For (III)$\Rightarrow$(II), we show that Prover's strategy described in
(III) is a winning strategy.  For (II)$\Rightarrow$(III), we show that every
winning strategy must satisfy the conditions of (III). For
(III)$\Rightarrow$(IV), we show that having vertices $u\prec v$ with
$\delta(u,v)\leq\beta(v)-2$ allows Adversary to win, by playing along the bad
walk defined by vertices $u,v$. Finally, for (IV)$\Rightarrow$(III), assuming no bad pair
$u,v$, we describe a Prover's strategy satisfying (III).

Before the proof of the remaining implications, we need to show some properties of the function
$\delta(\cdot,\cdot)$.

\begin{lemma}\label{lem:aux1}
Suppose that $G$ is a bipartite graph.  If $Q$ is a looping walk of
$G$ between $u$ and $v$, then\smallskip

\centereq{$\lambda(Q)+\delta(u,v)$ is an even number.}
\end{lemma}

\begin{proof}
Let $Q=x_1,\ldots,x_r$ be a looping walk of $G$ with $x_1=u$ and $x_r=v$.
By definition $\delta(u,v)\leq \lambda(Q)<\infty$. So there exists a looping
walk $Q'=x_1',\ldots,x_s'$ with $x_1'=u$ and $x_s'=v$ such that
$\delta(u,v)=\lambda(Q')$. We calculate:

\centereq{$\lambda(Q)+\delta(u,v)=\lambda(Q)+\lambda(Q')=
|Q|-2\Sum_{i=2}^{r-1}\big(\beta(x_i)-1\big)+|Q'|-2\Sum_{i=2}^{s-1}\big(\beta(x_i')-1\big)$}

\noindent It follows that $\lambda(Q)+\delta(u,v)$ is even if and only if
$|Q|+|Q'|$ is.  Recall that $x_1=x_1'=u$ and $x_r=x_s'=v$. Thus
$Q''=x_1,x_2,\ldots,x_{r-1},x'_s,x'_{s-1},\ldots,x'_1$ is a closed walk of
$G$ of length $|Q|+|Q'|$. Since $G$ is bipartite, it contains no
closed walk of odd length. Thus $|Q|+|Q'|$ is even, and so is
$\lambda(Q)+\delta(u,v)$.
\end{proof}

The following can be viewed as a triangle inequality for $\delta(\cdot,\cdot)$.

\begin{lemma}\label{lem:aux2}
If $u,v,w\in V(G)$ satisfy $u\neq v$ and $\{u,v\}\prec w$, then
$\delta(u,v)\leq \delta(u,w)+\delta(v,w)-2\beta(w)+2$.  Moreover, if $G$
is a bipartite graph, then $\delta(u,v)+\delta(u,w)+\delta(v,w)$ is an even
number or $\infty$.
\end{lemma}

\begin{proof}
If $\delta(u,w)=\infty$ or $\delta(v,w)=\infty$, the claim is clearly true.  So
we may assume that $\delta(u,w)<\infty$ and $\delta(v,w)<\infty$.  This means
that there exists a looping walk $Q=x_1,\ldots,x_r$ with $x_1=u$ and $x_r=w$
where $\lambda(Q)=\delta(u,w)$, and also a looping walk $Q'=x_1',\ldots,x_s'$
with $x_1'=v$ and $x_s'=w$ where $\delta(v,w)=\lambda(Q')$.\smallskip

Note that $x_r=x_s'=w$, and define $Q''=x_1, x_2, \ldots,x_{r-1}, x'_s, x'_{s-1},
\ldots, x_1'$.  We calculate:

\noindent$\lambda(Q)+\lambda(Q')=
|Q|-2\Sum_{i=2}^{r-1}\big(\beta(x_i)-1\big)+
|Q'|-2\Sum_{i=2}^{s-1}\big(\beta(x'_i)-1\big)$

\hspace*{1.4em}$=|Q|+|Q'|-2\Sum_{i=2}^{r-1}\big(\beta(x_i)-1\big)-2(\beta(w)-1)
-2\Sum_{i=2}^{s-1}\big(\beta(x'_i)-1\big)+2(\beta(w)-1)$\medskip

\hspace*{1.4em}$=\lambda(Q'')+2\beta(w)-2$\medskip

Observe that $Q''$ is a walk of $G$ whose endpoints are $u$ and $v$.  We
verify that $Q''$ is in fact a looping walk of~$G$.  We need to check the
conditions (i)-(iii) of the definition.  For (i), we recall the assumption
$u\neq v$.  For (ii), we note that $\{x_1,x_r\}\prec\{x_2,\ldots,x_{r-1}\}$
since $Q$ is a looping walk. Similarly, $\{x_1',x_s'\} \prec
\{x'_2,\ldots,x'_{s-1}\}$ since $Q'$ is a looping walk.  We also recall the
assumptions $\{u,v\}\prec w$. Thus, since $w=x_r=x'_s$, we conclude that
$\{u,v\}\prec w \prec\{x_2, \ldots, x_{r-1}, x'_2, \ldots, x'_{s-1}\}$ which
shows (ii). Finally, (iii) follows from the fact that both $Q$ and $Q'$ are
looping walks.  This verifies that $Q''$ indeed is a looping walk of $G$
between $u$ and $v$.

So we have $\delta(u,v)\leq \lambda(Q'')$ by definition, and we can calculate:\smallskip

\centereq{$\delta(u,w)+\delta(v,w)=\lambda(Q)+\lambda(Q')=
\lambda(Q'')+2\beta(w)-2\geq \delta(u,v)+2\beta(w)-2$} \medskip

\noindent Thus $\delta(u,v)\leq \delta(u,w)+\delta(v,w)-2\beta(w)+2$ as claimed.

Now, assume that $G$ is a bipartite graph.  Recall that $\lambda(Q)+
\lambda(Q')= \lambda(Q'')+2\beta(w)-2$. This implies that $\lambda(Q)+
\lambda(Q')+ \lambda(Q'')$ is an even number. By Lemma \ref{lem:aux1} also
$\lambda(Q)+ \delta(u,w)$ and $\lambda(Q')+ \delta(v,w)$ are even, since
$G$ is assumed to be bipartite.  For the same reason
$\lambda(Q'')+\delta(u,v)$ is even.  So since $\lambda(Q)+ \lambda(Q')+
\lambda(Q'')$ is even, it follows that $\delta(u,v)+ \delta(u,w)+ \delta(v,w)$
is even, as claimed.
\end{proof}

\begin{proof}[Theorem \ref{thm:12csp-path} (III)$\Rightarrow$(II)]
Assume (III), namely that Prover can play so as to satisfy ($\star$) and
($\triangle$) in every instance of the game. We show that this is a winning
strategy for Prover, thus proving (II). 

To this end, we need to verify that in every instance of the game the resulting
mapping $f$ is a homomorphism of $G$ to $\bP_\infty$.  Namely, we verify
that for every edge $uv\in E(G)$, the mapping $f$ satisfies
$|f(u)-f(v)|=1$.

Consider an edge $uv\in E(G)$. Observe that $Q=u,v$ is a looping walk of
$G$ with $\lambda(Q)=1$. Thus $\delta(u,v)\leq\lambda(Q)=1$ by the
definition of $\delta(u,v)$.  Using Lemma \ref{lem:aux1}, we deduce that
$\lambda(Q)+\delta(u,v)=1+\delta(u,v)$ is even.  Thus $\delta(u,v)$ is odd. By
$(\triangle)$, we have that $f(u)+f(v)+\delta(u,v)$ is even.  Thus $f(u)+f(v)$
is odd, since $\delta(u,v)$ is. Further, by ($\star$), we observe that
$|f(u)-f(v)|\leq \delta(u,v)$.  Thus $|f(u)-f(v)|\leq 1$, since
$\delta(u,v)\leq\lambda(Q)=1$. Also $|f(u)-f(v)|\geq 1$, since $|f(u)-f(v)|\geq
0$ and $f(u)+f(v)$ is odd.

Thus together we conclude that $|f(u)-f(v)|=1$ as required.
\end{proof}

\begin{proof}[Theorem \ref{thm:12csp-path} (II)$\Rightarrow$(III)]
Assume (II), i.e. Prover has a winning strategy. Assume that Prover plays this
strategy. Then no matter how Adversary plays, Prover always wins. We show that
this strategy satisfies the conditions of (III).

For contradiction, suppose that there is an instance of the game for which the
conditions of (III) do not hold. Let $g$ denote the resulting mapping
produced by this instance.  Let us play the game again to produce a new mapping
$f$ by making Adversary play according to the following rules:\medskip

\begin{compactenum}
\em
\item When a vertex $v$ is considered, check if the set $S_v$ that Prover offers
for $v$ contains $\gamma\in S_v$ for which there exists $u\prec v$ with
$\delta(u,v)<\infty$ such that $|f(u)-\gamma|>\delta(u,v)$ or such that
$f(u)+\gamma+\delta(u,v)$ is odd.\medskip

\item If such $\gamma$ exists, choose $f(v)=\gamma$. If such $\gamma$ does not
exist and $f(u)=g(u)$ for all $u\prec v$, then choose $f(v)=g(v)$. If neither is
possible, then choose any value from $S_v$ for $f(v)$.\medskip
\end{compactenum}

It follows that $f$ does not satisfy the conditions of (III) much like $g$.
Indeed, if in step 2.  we find that $\gamma$ exists, then setting $f(v)=\gamma$
makes ($\star$) or ($\triangle$) fail for the pair $u,v$. Thus $f$ fails the
conditions of (III).  Otherwise, if $\gamma$ in step 2. never exists, then we
conclude that $f=g$, and hence, $f$ again fails the conditions of (III), since
$g$ does. (For this to hold, it is important to note that this is only possible
because Prover plays the same deterministic strategy in both instances of the
game; thus Prover will offer the same values for $f(v)$ as she did for $g(v)$ as
long as Adversary makes the same choices for $f$ as he did for $g$; i.e., as
long as $f(u)=g(u)$ for all vertices $u\prec v$.) \medskip

Since Prover plays a winning strategy, we must conclude that $f$ is a
homomorphism of $G$ to $\bP_\infty$. Namely, we have that each $uv\in
E(G)$ satisfies $|f(u)-f(v)|=1$.  We show that this leads to a
contradiction.  From this we will conclude that $g$ does not exist, and hence,
Prover's winning strategy satisfies (III) as we advertised earlier.\medskip

We say that a vertex $v$ is {\em good} if for all $u\prec v$ with
$\delta(u,v)<\infty$ the conditions ($\star$) and ($\triangle$) hold.
Otherwise, we say that $v$ is {\em bad}.  The following is a restatement of
Adversary's strategy as described above. 
\medskip

\noindent({\it+)~
A vertex $v$ is good if and only if\\\hspace*{3em} every $\gamma\in S_v$ and every $u\prec v$
with $\delta(u,v)<\infty$ are such that\smallskip

\centereq{$|f(v)-\gamma|\leq \delta(u,v)$
\qquad and\qquad $f(v)+\gamma+\delta(u,v)$ is even.}}\medskip

\noindent
To see this, observe that if $v$ is bad, then, by definition, the statement on
the right fails for $\gamma=f(v)\in S_v$. Conversely, if for some $\gamma\in
S_v$ there is $u\prec v$ with $\delta(u,v)<\infty$ such that
$|f(v)-\gamma|>\delta(u,v)$ or such that $f(v)+\gamma+\delta(u,v)$ is odd,
Adversary chooses $f(v)=\gamma$ in step 2. Thus $v$ is not a good vertex.
This~proves~$(+)$.\medskip

Since $f$ fails the conditions of (III), there exists at least one bad
vertex. Among all bad vertices, choose $v$ to be the bad vertex that is largest
with respect to $\prec$. 

Since $v$ is bad, there exists $u\prec v$ with $\delta(u,v)<\infty$ such that
$|f(u)-f(v)|>\delta(u,v)$, or such that $f(u)+f(v)+\delta(u,v)$ is odd.  In
particular, since $\delta(u,v)<\infty$, there exists a looping walk
$Q=x_1,\ldots,x_r$ with $x_1=u$ and $x_r=v$ such that $\lambda(Q)=\delta(u,v)$.
Clearly, $r\geq 2$ by definition.

Suppose first that $r=2$. Then $Q=u,v$ and $\lambda(Q)=1$. In particular, $uv\in
E(G)$ and $\delta(u,v)=\lambda(Q)=1$.  Recall that we assume that
$|f(u)-f(v)|>\delta(u,v)$ or that $f(u)+f(v)+\delta(u,v)$ is odd.  If
$|f(u)-f(v)|>\delta(u,v)$, then $|f(u)-f(v)|>1$, since $\delta(u,v)=1$. But
then, since $uv\in E(G)$, we have that $f$ is not a homomorphism, a
contradiction. So $f(u)+f(v)+\delta(u,v)$ must be odd. Thus $f(u)+f(v)$ is even,
because $\delta(u,v)=1$. But then $|f(u)-f(v)|\neq 1$ again contradicting our
assumption that $f$ is a homomorphism. 

This excludes the case $r=2$.  Thus we may assume $r\geq 3$. Since $Q$ is a
looping walk, this implies that there exists $\ell\in\{2,\ldots,r-1\}$ such that
both $Q_1=x_1,\ldots,x_\ell$ and $Q_2=x_\ell,\ldots,x_r$ are looping walks of
$G$.

Let us denote $w=x_\ell$.  So $Q_1$ is a looping walk from $u$ to $w$, while
$Q_2$ is a looping walk from $w$ to $v$. This implies that
$\delta(u,w)\leq\lambda(Q_1)$ and $\delta(v,w)\leq\lambda(Q_2)$.  Note that
$u\prec v\prec w$, since $\{u,v\}\prec\{x_2,\ldots,x_{r-1}\}$ because $Q$ is a
looping walk.  Thus by the maximality of $v$, we deduce that $w$ is a good
vertex. We calculate:\smallskip

\noindent$\lambda(Q)= |Q|-2\Sum_{i=2}^{r-1}\big(\beta(x_i)-1\big)$

$= |Q_1|+ |Q_2|-
2\Sum_{i=2}^{\ell-1}\big(\beta(x_i)-1\big)- 2(\beta(w)-1)-
2\Sum_{i=\ell+1}^{r-1}\big(\beta(x_i)-1\big)$ 
\smallskip

$=\lambda(Q_1)+\lambda(Q_2)-2\beta(w)+2\geq
\delta(u,w)+\delta(v,w)-2\beta(w)+2\geq \delta(u,v) =\lambda(Q)$\medskip

\noindent The last inequality is by Lemma \ref{lem:aux2}.  Therefore, we
conclude that $\delta(u,v)=\delta(u,w)+\delta(v,w)-2\beta(w)+2$.

Recall that we assume that $|f(u)-f(v)|>\delta(u,v)$ or that
$f(u)+f(v)+\delta(u,v)$ is odd. We can exclude the latter as follows.  Since $w$
is a good vertex, both $f(u)+f(w)+\delta(u,w)$ and $f(v)+f(w)+\delta(v,w)$ are
even by ($\triangle$).  This implies that $f(u)+f(v)+\delta(u,w)+\delta(v,w)$ is
even.  Therefore, $f(u)+f(v)+\delta(u,v)$ is even, since
$\delta(u,v)=\delta(u,w)+\delta(v,w)-2\beta(w)+2$.  In other words,
$f(u)+f(v)+\delta(u,v)$ is not odd, so the remaining possibility, by our
assumptions, is that $|f(u)-f(v)|>\delta(u,v)$. In particular, since
$f(u)+f(v)+\delta(u,v)$ is even, it follows that $|f(u)-f(v)|\neq\delta(u,v)+1$.
Thus $|f(u)-f(v)|\geq \delta(u,v)+2$. Note that $\beta(w)\in\{1,2\}$.

Now, recall that $w$ is a good vertex. Thus by $(+)$ we have that every
$\gamma\in S_w$ satisfies $|f(u)-\gamma|\leq \delta(u,w)$ and
$|f(v)-\gamma|\leq\delta(v,w)$. Using this we can calculate:\medskip

$|f(u)-f(v)|\geq \delta(u,v)+2=\delta(u,w)+\delta(v,w)-2\beta(w)+2+2$\medskip

\hspace*{6em}$\geq |f(u)-\gamma|+|f(v)-\gamma|-2\beta(w)+4\geq |f(u)-f(v)|-2\beta(w)+4$
\medskip

\noindent This implies that $\beta(w)=2$  and all the above inequalities are in
fact equalities.  Namely, the set $S_w$ contains distinct values
$\gamma_1,\gamma_2$ which satisfy $|f(u)-\gamma_1|=\delta(u,w)=|f(u)-\gamma_2|$
and $|f(v)-\gamma_1|=\delta(v,w)=|f(v)-\gamma_2|$.  Since
$\gamma_1\neq\gamma_2$, it follows that $\delta(u,w)\geq 1$ and
$f(u)-\gamma_1=-f(u)+\gamma_2$.  Similarly, we deduce $\delta(v,w)\geq 1$ and
$f(v)-\gamma_1=-f(v)+\gamma_2$. Thus $2f(u)=\gamma_1+\gamma_2=2f(v)$ which
yields $f(u)=f(v)$.  We again calculate:\medskip

\centereq{$|f(u)-f(v)|\geq \delta(u,v)+2 = \delta(u,w) + \delta(v,w) - 2\beta(w)
+ 2 + 2\geq 1 + 1 - 2\cdot 2 + 2 + 2=2$}\medskip

\noindent But this is clearly impossible, since $f(u)=f(v)$.

The proof is now complete.
\end{proof}

\begin{proof}[Theorem \ref{thm:12csp-path} (III)$\Rightarrow$(IV)]
We prove the contrapositive.  Assume that (IV) fails. Namely, suppose that there
are $u,v\in V(G)$ with $u\prec v$ such that $\delta(u,v)\leq\beta(v)-2$.
We show that Adversary can play to violate $(\star)$ for $u,v$. This will imply
that (III) fails.  If $\delta(u,v)\leq -1$, then $(\star)$ can never be
satisfied, since $|f(u)-f(v)|$ is always non-negative.  In that case (III) fails
no matter how Adversary plays.  So we may assume that $\delta(u,v)\geq 0$.  This
implies that $\delta(u,v)=0$ and $\beta(v)=2$, since we assume $\delta(u,v)\leq
\beta(v)-2$.  So Prover must offer to Adversary two distinct values for $v$.
Since the values are different, at least one of them must be different
from~$f(u)$.  Adversary chooses this value for $f(v)$.  This yields
$|f(u)-f(v)|\geq 1$ which violates $(\star)$, since $\delta(u,v)=0$.

This shows that Prover cannot play to always satisfy the conditions of (III),
and hence (III) fails as claimed.
\end{proof}

\begin{proof}[Theorem \ref{thm:12csp-path} (IV)$\Rightarrow$(III)]
Assume (IV), namely that all $u,v\in V(G)$ with $u\prec v$ satisfy
$\delta(u,v)\geq \beta(v)-1$.  We explain how Prover can play so as to satisfy
the conditions of (III).  We proceed by induction on the number of processed
vertices during the game.  As usual let $f$ denote the (partial) assignment
constructed by Adversary.

Consider a vertex $v\in V(G)$ and assume that $(\star)$ and $(\triangle)$
hold for all $u,w\in V(G)$ such that $\{u,w\}\prec v$.  We show how to
construct a set $S_v$ so that whatever the choice Adversary makes from this set
for $f(v)$, it will satisfy $(\star)$ and $(\triangle)$ for $v$ and all $u\in
V(G)$ with $u\prec v$. This will yield a strategy (for Prover) satisfying
(III).  Note that $\delta(u,v)\geq 0$ for all $u\prec v$, since
$\delta(u,v)\geq\beta(v)-1$ by our assumption, and $\beta(v)\in\{1,2\}$.

If $v$ is the first vertex in $\prec$ or if every $u\prec v$ is such that
$\delta(u,v)=\infty$, then there is no condition we need to satisfy; Prover
simply offers any set $S_v\subseteq\mathbb{Z}$.  So we may assume that this is
not the case. 

For each $u\prec v$ with $\delta(u,v)<\infty$, we write $\I_u$ to denote the
closed interval defined as follows.  \smallskip

\centereq{$\I_u=\Big[f(u)-\delta(u,v),\,f(u)+\delta(u,v)\Big]$}
\smallskip

Let $\L$ be the intersection of all intervals $\I_u$ where $u\prec v$ and
$\delta(u,v)<\infty$. Observe that in order to satisfy ($\star$) we must choose
$S_v$ to be a subset of $\L$. We show that this is possible while also
satisfying $(\triangle)$.

Since $\L$ is the intersection of intervals, there exist $x,y$ such that
$\L=\I_x\cap \I_y$. Note that $\{x,y\}\prec v$ and $\delta(x,v)<\infty$ as well
as $\delta(y,v)<\infty$. By symmetry, we may assume that $f(x)+\delta(x,v)\leq
f(y)+\delta(y,v)$.

We prove that $|\L|\geq 2\beta(v)-1$.  If $\I_x\subseteq\I_y$, then we have
$|\L|=|\I_x\cap \I_y|=|\I_x|$ and we recall that we assume
$\delta(x,v)\geq\beta(v)-1$, since $x\prec v$.  Thus $|\I_x|=2\delta(x,v)+1\geq
2\beta(v)-1$ and so $|\L|\geq 2\beta(v)-1$ as claimed.  Similarly, if
$\I_y\subseteq\I_x$, we conclude that $|\L|=|\I_y|$ and $|\I_y|\geq 2\beta(v)-1$
by our assumption.  

Thus we may assume that $\I_x\not\subseteq \I_y$ and $\I_y\not\subseteq\I_x$.
Note that this also implies that $x\neq y$.  Recall that $f(x)+\delta(x,v)\leq
f(y)+\delta(y,v)$. We claim that also $f(x)<f(y)$.  Indeed, if $f(y)\leq f(x)$,
then $f(x)+\delta(x,v)\leq f(y)+\delta(y,v)\leq f(x)+\delta(y,v)$ which yields
$\delta(x,v)\leq \delta(y,v)$. Thus $f(y)-\delta(y,v)\leq f(y)-\delta(x,v) \leq
f(x)-\delta(x,v)$. But this leads to $\I_x\subseteq\I_y$, since we also assume
$f(x)+\delta(x,v)\leq f(y)+\delta(y,v)$.

So we conclude that $f(x)<f(y)$ from which we deduce that $\I_x\cap\I_y =
\big[f(y)-\delta(y,v), \,f(x) + \delta(x,v)\big]$.  Thus $|\I_x\cap\I_y| = f(x)
- f(y) + \delta(y,v) + \delta(x,v) + 1$. Since $f(x)<f(y)$, we can express this
as\medskip

\centereq{$|\I_x\cap\I_y|=\delta(x,v)+\delta(y,v)-|f(x)-f(y)|+1$}\medskip

\noindent Recall further that $\{x,y\}\prec v$ and $x\neq y$. Thus
$\delta(x,y)\leq \delta(x,v)+\delta(y,v)-2\beta(v)+2$ by Lemma \ref{lem:aux2}.
Also, note that $|f(x)-f(y)|\leq\delta(x,y)$,  since we assume that $(\star)$
holds for $x,y$, because $\{x,y\}\prec v$.  Therefore\medskip

$|\L|=|\I_x\cap\I_y|=\delta(x,v)+\delta(y,v)-|f(x)-f(y)|+1$\smallskip

\hspace*{1.5em}$\geq
\delta(x,y)+2\beta(v)-2-\delta(x,y)+1=2\beta(v)-1$\medskip

\noindent This proves that indeed $|\L|\geq 2\beta(v)-1$.

We are finally ready to construct the set $S_v$.
Recall that $\L=\I_x\cap\I_y$ and we assume that $f(x)+\delta(x,v)\leq
f(y)+\delta(y,v)$.  This implies that the value\medskip

\centereq{$\gamma=f(x)+\delta(x,v)$}\medskip

\noindent is the largest element of $\L$.  We claim that $f(v)=\gamma$ satisfies
($\triangle)$.  Suppose that ($\triangle$) fails for $f(v)=\gamma$ and some
$u\prec v$ with $\delta(u,v)<\infty$.  Namely, suppose that
$f(u)+\gamma+\delta(u,v)$ is odd.  Thus $f(u)+f(x)+\delta(x,v)+\delta(u,v)$ is
odd.  This implies that $u\neq x$. Note that $\{x,u\}\prec v$. Thus from Lemma
\ref{lem:aux2}, we obtain that $\delta(x,u)+\delta(x,v)+\delta(u,v)$ is even.
Thus $f(u)+f(x)+\delta(x,u)$ is odd, since $f(u)+f(x)+\delta(x,v)+\delta(u,v)$
is. Hence, $(\triangle)$ fails for $x,u$. However, we assume that $(\triangle)$
holds for $x,u$, since $\{x,u\}\prec v$, a contradiction.  So no such a $u$
exists. 

Therefore $(\triangle)$ holds for $f(v)=\gamma$ and every $u\prec v$ with
$\delta(u,v)<\infty$.  Thus $(\triangle)$ also clearly holds for
$f(v)=\gamma-2$.  In particular, if $\beta(v)=2$, then $\gamma-2\in \L$.  This
follows from the fact that $|\L|\geq 2\beta(v)-1$ and $\gamma$ is the largest
element in $\L$.  Finally, recall that each element in $\L$ can be used as
$f(v)$ to satisfy $(\star)$. In particular, $(\star)$ holds for $f(v)=\gamma$
and if $\beta(v)=2$, it also holds for $f(v)=\gamma-2$, since in that case
$\gamma-2\in\L$.

This together shows that if $\beta(v)=1$, Prover can safely offer
$S_v=\{\gamma\}$ while if $\beta(v)=2$, Prover can safely offer
$S_v=\{\gamma,\gamma-2\}$. Any choice Adversary makes for $f(v)$ from such $S_v$
is guaranteed to make $(\star)$ and $(\triangle)$ hold for all $u\prec v$ with
$\delta(u,v)<\infty$.  This is precisely what we set out to prove.

The proof is now complete.
\end{proof}

With this characterization, the proof of Theorem \ref{thm:inf-path} is
straightforward.

\subsection{Proof of Theorem \ref{thm:inf-path}}
We observe that the values $\delta(u,v)$ can be
easily computed in polynomial time by dynamic programming. This allows us to
test conditions of Theorem \ref{thm:12csp-path} and thus decide $\{1,2\}$-CSP($\P_\infty$)
in polynomial time as claimed.\hfill\qed

\section{Algorithm for $\{1,2\}$-CSP($\bP_n$)}\label{sec:finite-path}

The path $\bP_n$ has vertices $\{1,2,\ldots,n\}$ and edges $\{ij~:~|i-j|=1\}$.
Let $\Psi$ be an instance of $\{1,2\}$-CSP($\bP_n$).  As usual, let $G$ be the
graph $\D_\psi$ corresponding to $\Psi$, and let $\prec$ be the corresponding
total ordering of $V(G)$.

For simplicity, we shall assume that $G$ is connected and bipartite with white
and black vertices forming the bipartition. (If it is not bipartite, there is no
solution; if disconnected, we solve the problem independently on each
component.) No generality is lost this way.

We start by characterizing small cases as follows.

\begin{lemma} Assume $P_\infty\models\Psi$. Let $f$ be the first vertex in the
ordering $\prec$. Then
\begin{compactenum}[(i)]
\item
$\bP_1 \models \Psi$   $\iff$ $G$ is the single $\exists^{\geq 1}$ vertex $f$.

\item
$\bP_2 \models \Psi$   $\iff$ $G$ does not contain $\exists^{\geq 2}$ vertex
except possibly for $f$.

\item
$\bP_3 \models \Psi$ $\iff$ all $\exists^{\geq 2}$ vertices in $G$ have the same
colour.

\item
$\bP_4 \models \Psi$ $\iff$ all $\exists^{\geq 2}$ vertices in $G$ are pairwise
non-adjacent except possibly for $f$

\item
$\bP_5 \models \Psi$  $\iff$ there is colour $C$ (black or white) such that each
edge xy between two $\exists^{\geq 2}$ vertices where $x\prec y$ is such that
$x$ has colour $C$
\end{compactenum}
\end{lemma}

\begin{proof}
(i) is clear. For (ii), any $\exists^{\geq 2}$ vertex other than $f$ must be
offered both 1 and 2, one of which will violate the parity with respect to $f$
(since $G$ is connected).  In all other cases, $\bP_2\models\Psi$ because
$\D_\psi$ is bipartite.

Similar argument works for (iii), if there are two $\exists^{\geq 2}$ vertices $u,v$ of
different colour where $u\prec v$, then Prover must offer 1 or 3 among the
values for $u$ and Adversary chooses it. She also must offer 1 or 3 among the
values for $v$ and again, Adversary chooses it.  This now violates the parity,
since $u$ and $v$ are in different sides of the bipartition of $G$.  Conversely,
if all $\exists^{\geq 2}$ vertices are, say black, then Prover offers $\{1,3\}$
to all black $\exists^{\geq 2}$ vertices, 1 to all black $\exists^{\geq 1}$
vertices, and $2$ to all white vertices. This will allow Prover to win.

For (iv), if $G$ contains adjacent $\exists^{\geq 2}$ vertices $u,v$ distinct
from $f$ where $u\prec v$, then Prover must offer $\{1,3\}$ or $\{2,4\}$ for $u$
because of the parity with respect to $f$. Adversary chooses either 1 or 4 and
Prover subsequently loses at $v$.  Conversely, if no such vertices $u,v$ exist,
Prover first offers $\{2,3\}$ for $f$.  By symmetry of the path $\bP_4$, we may
assume that Advesary chooses $2$ for $f$, and that $f$ is black.  Prover
subsequently offers 2 for each black $\exists^{\geq 1}$ vertex, and offers $3$
for each white $\exists^{\geq 1}$ vertex. She offers $\{2,4\}$ for each
$\exists^{\geq 2}$ black vertex and offers $\{1,3\}$ for each white
$\exists^{\geq 2}$ vertex. Since no two $\exists^{\geq 2}$ vertices are adjacent, any
Adversary's choices lead to a homomorphism, and so Prover always wins.

Finally, for (v), by symmetry, assume that $f$ is black. Suppose that $G$
contains edges $xy$ and $wz$ where $x,y,w,z$ are $\exists^{\geq 2}$ vertices
with $x\prec y$ and $w\prec z$, and where $x$ is black and $w$ is white
(possibly $f=x$).  If $2$ or $4$ is chosen for $f$, then $\{1,3\}$ or $\{1,5\}$
or $\{3,5\}$ is offered for $w$ because of the parity ($f$ is black and $w$ is
white). This allows Adversary to choose 1 or 5 for $w$ and Prover loses at $z$.
If one of $1,3,5$ is chosen for $f$, then 1 or 5 is among values offered for $x$
because of parity. So Adversary can choose 1 or 5 for $x$ and Prover loses at
$y$.

Conversely, suppose first that every edge $xy$ where $x,y$ are $\exists^{\geq
2}$ vertices is such that $x$ is black. Prover offers $\{2,4\}$ for $f$ and for
every black $\exists^{\geq 2}$ vertex. She offers 3 for each white
$\exists^{\geq 1}$ vertex. Each white $\exists^{\geq 2}$ vertex will be offered
$\{1,3\}$ if 2 was chosen for all its predecessors, or $\{3,5\}$ if 4 was chosen
on all its predecessors. This is always possible as guaranteed by Theorem
\ref{thm:12csp-path}(III) and (IV), since otherwise $\bP_\infty$ does not model
$\Psi$, let alone $\bP_5$. By the same token, it follows that no black
$\exists^{\geq 1}$ vertex will have two predecessors for which 1 and 5
respectively was chosen.  So one of 2,4 can be always successfully offered for
each black $\exists^{\geq 1}$ vertex. This shows that Prover indeed always wins.

Thus we may assume that every edge $xy$ where $x,y$ are $\exists^{\geq 2}$
vertices is such that $x$ is white. Here Prover offers $3$ or $\{3,5\}$ for $f$.
Then we proceed exactly as in the previous case, just switching colours black and
white.
\end{proof}

We now expand this lemma to the general case of $\{1,2\}$-CSP($\bP_n$) as
follows.  Recall that we proved that $\bP_\infty\models\Psi$ if and only if Prover
can play $\G(\Psi,\bP_n)$ so that in every instance of the game, the resulting
mapping $f$ satisfies ($\star$) and ($\triangle$). In fact the proof of
(III)$\Rightarrow$(II) from Theorem \ref{thm:12csp-path} shows that every
winning strategy of Prover has this property. We use this fact in the subsequent
text. 

We separately investigate the case when $n$ is odd and when $n$ is even.

\subsection{Even case}
In the following definition, we describe for each vertex $v$ the value
$\gamma(v)$ using a recursive definition. We shall use this value to keep track
of the distance of $f(v)$ from the center of the path $\bP_n$.

\begin{definition} 
For each vertex $v$ we define $\gamma(v)$ recursively as follows:\smallskip

\hspace*{2.5em}$\gamma(v)=0$\qquad if $v$ is first in the ordering $\prec$

else\quad $\gamma(v)=\beta(v)-1+\max\bigg\{0,\,\displaystyle\max_{u\prec
v}\Big(\gamma(u)-\delta(u,v)+\beta(v)-1\Big)\,\bigg\}$
\end{definition}

\begin{lemma}\label{lem:path1}
Let $M$ be real number. Suppose that $\bP_\infty\models \Psi$ and that Prover plays a
winning strategy in the game $\G(\Psi,\bP_\infty)$. Then Adversary can play so that
the resulting mapping $f$ satisfies $|f(v)-M|\geq \gamma(v)$ for every vertex
$v\in V(D_\psi)$.
\end{lemma}

\begin{proof}
As we remarked above the definitions, since Prover plays a winning strategy, she
must satisfy $(\star)$ and $(\triangle)$.  Adversary when given choice between
two values will choose the value that is farther away from $M$ (ties broken
arbitrarily).

We prove the claim by induction on the number of steps.  Consider some step in
the game, and let $v$ denote the vertex considered in this step.  
If $v$ is the first vertex in $\prec$, then $\gamma(v)=0$ and
$|f(v)-M|\geq 0=\gamma(v)$ holds. 
So we may assume that $v$ is not first in $\prec$.  Thus

\mbox{}\hfill$\gamma(v)=\beta(v)-1+\max\Big\{0,\displaystyle\max_{u\prec
v}\big(\gamma(u)-\delta(u,v)+\beta(v)-1\big)\Big\}$\hfill$(1)$
\smallskip

Suppose that no vertex maximizes $\gamma(v)$, i.e.  we have
$\gamma(v)=\beta(v)-1$.  If $\beta(v)=1$, then $\gamma(v)=\beta(v)-1=0$ and so
again $|f(v)-M|\geq 0=\gamma(v)$. 

So we may assume that $\beta(v)=2$ in which case Prover offers two distinct
values $\alpha_1,\alpha_2$.  By symmetry, we may assume
$|M-\alpha_1|\geq|M-\alpha_2|$. Recall that $v$ is not first in $\prec$ and $\D_\psi$
is connected. Since the values that Prover offers must satisfy $(\triangle)$
when used as $f(v)$, we deduce that $\alpha_1,\alpha_2$ must have the same
parity. Indeed, since $\D_\psi$ is connected and $v$ is not first in $\prec$, we deduce
that there exists a looping walk from $u$ to $v$ where $u\prec v$. Then
$(\triangle)$ applied to $u$ and $v$ implies that the values $\alpha_1,\alpha_2$
have the same parity. Thus $|\alpha_1-\alpha_2|\geq 2$, and since
$|M-\alpha_1|\geq |M-\alpha_2|$, we deduce that\smallskip

\centereq{$2\leq |\alpha_1-\alpha_2|\leq |M-\alpha_1|+|M-\alpha_2|\leq 2|M-\alpha_1|$}
\smallskip

Thus $|M-\alpha_1|\geq 1$ and Adversary chooses $f(v)=\alpha_1$.
So, we conclude that $|f(v)-M|\geq 1=\beta(v)-1=\gamma(v)$ as claimed.
\smallskip

Now, we may assume that there is a vertex $u\prec v$ that maximizes
$\gamma(v)$, namely $\gamma(v)=2(\beta(v)-1)+\gamma(u)-\delta(u,v)$.
From the inductive hypothesis, we know that $|f(u)-M|\geq\gamma(u)$. Without
loss of generality, assume that $f(u)\geq M$.
From this we calculate
\smallskip

\centereq{$f(u)-M=|f(u)-M|\geq\gamma(u)=\gamma(v)+\delta(u,v)-2(\beta(v)-1)$}\smallskip

Now, recall that Prover offers values that satisfy $(\star)$ when chosen for $f(v)$.
There are two cases.
\smallskip

\noindent{\bf Case 1:}
assume $\beta(v)=1$. Then Prover offers one value that becomes
$f(v)$ where $|f(u)-f(v)|\leq\delta(u,v)$ by~$(\star)$. In other words,
$f(u)-\delta(u,v)\leq f(v)\leq f(u)+\delta(u,v)$ and so we rewrite (since
$\beta(v)-1=0$)\medskip

\centereq{$f(v)\geq f(u)-\delta(u,v)
\geq f(u)-\delta(u,v)+2(\beta(v)-1)
\geq M+\gamma(v)$}\medskip

\noindent from which we conclude $f(v)-M\geq \gamma(v)$ and therefore $|f(v)-M|\geq \gamma(v)$.
\medskip

\noindent{\bf Case 2:} assume $\beta(v)=2$. Then Prover offers two values
$\alpha_1,\alpha_2$ where $\alpha_1< \alpha_2$, and both satisfy ($\star$) and
($\triangle$) in place of $f(v)$.  Namely, for $i=1,2$, we have
$f(u)-\delta(u,v)\leq \alpha_i\leq f(u)+\delta(u,v)$.  From ($\triangle$) we
know that $\alpha_1,\alpha_2$ have the same parity. Thus since
$\alpha_1<\alpha_2$, we deduce $\alpha_1+2\leq \alpha_2$. Thus we can can write
$f(u)-\delta(u,v)+2\leq \alpha_1+2\leq \alpha_2$.  Adversary chooses
$f(v)=\alpha_2$ which yields\medskip

\centereq{$f(v)=\alpha_2\geq f(u)-\delta(u,v)+2=f(u)-\delta(u,v)+2(\beta(v)-1)\geq
M+\gamma(v)$}\medskip

\noindent from which we again conclude $f(v)-M\geq \gamma(v)$ and hence
$|f(v)-M|\geq\gamma(v)$.

This completes the proof.
\end{proof}

\begin{lemma}\label{lem:path2}
Let $M$ be a real number. Suppose that $\bP_\infty\models\Psi$. Then there exists
a winning strategy for Prover such that in every instance of the game the
resulting mapping $f$ satisfies $|f(v)-M|\leq\gamma(v)+1$ for every $v\in
V(\D_\psi)$.
\end{lemma}

\begin{proof}
In the desired strategy, Prover offers values closest to $M$.  We prove the
claim by induction on the number of steps.  Since $\bP_\infty\models\Psi$, we
conclude that the condition (IV) of Theorem \ref{thm:12csp-path} holds. This
allows us to follow the induction in the proof of implication
(IV)$\Rightarrow$(III) of said theorem.

In this proof, it is shown that for each vertex $v$, the set of feasible choices
for $f(v)$ forms an interval ${\cal L}$ where either ${\cal L}=\mathbb{R}$ or
${\cal L}=[f(y)-\delta(y,v),\,f(x)+\delta(x,v)]$ where $\{x,y\}\prec v$,
possibly $x=y$. Since (IV) holds, this interval contains at least $2\beta(v)-1$
elements and to satisfy ($\triangle$) we may offer for $v$ the value
$f(x)+\delta(x,v)$ or also $f(x)+\delta(x,v)-2$ if needed.  By the same token we
may offer $f(y)-\delta(y,v)$ or $f(y)-\delta(y,v)+2$, or actually any other
value in $\cal L$ of the right parity.
\medskip

\noindent{\bf Case 1:} Suppose that $M\in{\cal L}$. If $v$ is the first vertex
in $\prec$, then we have Prover offer for $v$ the value $\lfloor M\rfloor$ if
$\beta(v)=1$, and have her offer values $\lfloor M\rfloor$, $\lfloor M\rfloor+1$
if $\beta(v)=2$. Adversary chooses $f(v)$ from these values and so $|f(v)-M|\leq
1=\gamma(v)+1$ as required.

So we may assume that $v$ is not first in $\prec$. If $\beta(v)=1$, we have
Prover offer for $v$ the value $\lfloor M\rfloor$ or $\lfloor M\rfloor+1$. We
choose one of  $\lfloor M\rfloor$, $\lfloor M\rfloor+1$ that has the right parity so
that $(\triangle)$ holds for all $u\prec v$ when this value is assigned to
$f(v)$. (Since $\D_\psi$ is bipartite, if $(\triangle)$ holds for one $u$, it holds
for all $u$; the choice is always~possible).

Similarly, if $\beta(v)=2$, then Prover offers either values $\lfloor
M\rfloor-1$, $\lfloor M\rfloor+1$, or values $\lfloor M\rfloor$, $\lfloor
M\rfloor+2$.  One of the two will satisfy ($\triangle$) for any subsequent choice
that Adversary will make. Moreover, since the endpoints of the interval~$\cal L$
(if ${\cal L}\neq\mathbb{R}$) when chosen for $f(v)$ satisfy ($\triangle$), it
also follows that the values Prover offers all belong to $\cal L$. This means that
any one of these values satisfies $(\star)$ when chosen for $f(v)$.

This shows that the value $f(v)$ that Adversary chooses 
satisfies $|f(v)-M|\leq\beta(v)$ and
$(\triangle)$,$(\star)$ hold.
Recall that

\mbox{}\hfill$\gamma(v)=\beta(v)-1+\max\Big\{0,\displaystyle\max_{u\prec
v}\big(\gamma(u)-\delta(u,v)+\beta(v)-1\big)\Big\}$\hfill$(1)$
\smallskip

\noindent This implies that $\gamma(v)\geq\beta(v)-1$. Thus $|f(v)-M|\leq
\beta(v)\leq \gamma(v)+1$ as required.\medskip

\noindent{\bf Case 2:}
Suppose that $M\not\in{\cal L}$.  Thus ${\cal L}\neq\mathbb{R}$ and either
$M<f(y)-\delta(y,v)$ or $f(x)+\delta(x,v)<M$. 
Note that this means that $v$ is not first in $\prec$.

By symmetry, we may assume that
$M>f(x)+\delta(x,v)$.
If $\beta(v)=1$, we have Prover offer the value $f(x)+\delta(x,v)$.
If $\beta(v)=2$, Prover offers values $f(x)+\delta(x,v)$, $f(x)+\delta(x,v)-2$.
This implies that the value $f(v)$ that Adversary chooses for $v$ satisfies
$(\star)$, $(\triangle)$ since $|{\cal L}|\geq 2\beta(v)-1$, and we have
\smallskip

\centereq{$f(x)+\delta(x,v)-2(\beta(v)-1)\leq f(v)\leq f(x)+\delta(x,v)$}
\smallskip

Thus since $f(x)+\delta(x,v)<M$, we deduce that $f(v)<M$. Also $f(x)<M$ since
$\delta(x,v)\geq 0$ by (IV). Applying the inductive hypothesis to $x$ this
yields
\smallskip

\centereq{$M-f(x)=|f(x)-M|\leq\gamma(x)+1$}
\smallskip

We now recall (1) and thus conclude (when using $u:=x$ in (1))
\smallskip

\centereq{$\gamma(v)\geq 2(\beta(v)-1)+\gamma(x)-\delta(x,v)$}
\smallskip

Finally, above we remarked that $f(v)<M$ and so we write
\medskip

$|f(v)-M|=M-f(v)\leq M-\big(f(x)+\delta(x,v)-2(\beta(v)-1)\big)$\smallskip

\hspace*{10.5em}$\leq(\gamma(x)+1)-\delta(x,v)+2(\beta(v)-1)$\smallskip

\hspace*{10.5em}$\leq\gamma(v)+1$
\medskip

This proves $|f(v)-M|\leq\gamma(v)+1$ as required. Analogous proof works in the case
$M<f(y)-\delta(y,v)$ where Prover offers $f(y)-\delta(y,v)$ or
$f(y)-\delta(y,v)+2$.

That completes the proof.
\end{proof}

With the above lemmas, we are now ready to characterize $\{1,2\}$-CSP($\bP_n$)
for even $n$.

\begin{theorem}\label{thm:npath-even} Let $n\geq 4$ be even. Assume that
$\bP_\infty\models \Psi$. Then the following are equivalent.
\begin{compactenum}[(I)]
\item $\bP_n\models \Psi$.
\item Prover has a winning strategy in the game $\G(\Psi,\bP_n)$.
\item There is no vertex $v$ with $\gamma(v)\geq \frac{n}{2}$.
\end{compactenum}
\end{theorem}

\begin{proof}
Note first that since $n$ is even, we may assume, without loss of generality,
the first vertex in the ordering is quantified $\exists^{\geq 1}$. If not, we
can freely change its quantifier to $\exists^{\geq 1}$ without affecting the
satisfiability of the instance.  Namely, for each value $i$, we can offer a pair
of symmetric values $i$, $n-i+1$; if Adversary chooses $n-i+1$, we simply start
offering values $n-j+1$ where we would offer $j$.

(I)$\Leftrightarrow$(II) is by Lemma \ref{lem:game}.  For
(II)$\Rightarrow$(III), assume there is $v$ with $\gamma(v)\geq \frac{n}{2}$ and
Prover has a winning strategy in $\mathscr G(\Psi,\bP_n)$. This is also a winning
strategy in $\mathscr G(\Psi,\bP_\infty)$. This allows us to apply Lemma
\ref{lem:path1} for $M=\frac{n+1}{2}$ to conclude that Adversary can play
against Prover so that $|f(v)-\frac{n+1}{2}|=|f(v)-M|\geq\gamma(v)=\frac{n}{2}$.
Thus either $f(v)\geq \frac{2n+1}{2}>n$ or $f(v)\leq \frac{1}{2}<1$. But then
$f(v)\not\in\{1,\ldots,n\}$ contradicting our assumption that Prover plays a
winning strategy.

For (III)$\Rightarrow$(II), assume that $\gamma(v)\leq\frac{n}{2}-1$ for all
vertices $v$. We apply Lemma \ref{lem:path2} for $M=\frac{n+1}{2}$. This tells
us that Prover has a winning strategy on $\G(\Psi,\bP_\infty)$ such that in every
instance of the game, if $f$ is the resulting mapping, the mapping satisfies
$|f(v)-\frac{n+1}{2}|\leq\gamma(v)+1$ for every vertex $v$.  From this we conclude
that $f(v)\geq\frac{n+1}{2}-\gamma(v)+1\geq \frac{n+1}{2}-\frac{n}{2}=\frac{1}{2}$
and that $f(v)\leq\frac{2n+1}{2}=n+\frac{1}{2}$.  Therefore
$f(v)\in\{1,2,\ldots,n\}$ confirming that $f$ is a valid homomorphism to $\bP_n$.
\end{proof}

\subsection{Odd case}

For odd $n$, we proceed similarly as for even $n$ except for a subtle twist. We
define $\gamma'(v)$ using same recursion as  $\gamma(v)$ except that we set
$\gamma'(v)=\beta(v)-1$ if $v$ is first in $\prec$. Namely, $\gamma'(v)$ is
defined as follows.

\begin{definition} 
For each vertex $v$ we define $\gamma'(v)$ recursively as follows:\smallskip

\centereq{$\gamma'(v)=\beta(v)-1+\max\bigg\{0,\,\displaystyle\max_{u\prec
v}\Big(\gamma'(u)-\delta(u,v)+\beta(v)-1\Big)\,\bigg\}$}
\end{definition}

\begin{lemma}\label{lem:odd-path1}
Let $M$ be an integer. Suppose that $\bP_\infty\models \Psi$ and that Prover plays a
winning strategy in the game $\G(\Psi,\bP_\infty)$. Then Adversary can play so that
the resulting mapping $f$ satisfies $|f(v)-M|\geq \gamma'(v)$ for every vertex
$v\in V(D_\psi)$.
\end{lemma}

\begin{proof}
The proof is by induction on the number of steps. For the base case, $v$ is
first in the ordering $\prec$ and we have $\gamma'(v)=\beta(v)-1$. If
$\beta(v)=1$, then $\gamma'(v)=0$ and $|f(v)-M|\geq 0=\gamma'(v)$.  If
$\beta(v)=2$, Prover offers two distinct values for $v$. At least one of them is
not $M$ and Adversary chooses this value as $f(v)$. So $f(v)\neq M$ and since
$M$ is an integer, we conclude that $|f(v)-M|\geq 1=\gamma'(v)$.

The rest of the proof (the inductive case) is exactly as in the proof of Lemma
\ref{lem:path1} (with $\gamma$ replaced by $\gamma'$). That completes the proof.
\end{proof}

\begin{lemma}\label{lem:odd-path2}
Let $M$ be an integer and let $t\in\{0,1\}$. Suppose that
$\bP_\infty\models\Psi$. Then there exists a winning strategy for Prover such
that in every instance of the game the resulting mapping $f$ satisfies
$|f(v)-M|\leq\gamma'(v)+1$ for every $v\in V(\D_\psi)$, and moreover,
$f(z)\equiv t~({\rm mod}~2)$ where $z$ is the first vertex in $\prec$.
\end{lemma}

\begin{proof}
By induction on the number of steps.  In the base case, we consider $z$, the
first vertex in $\prec$.  Recall that $\gamma'(z)=\beta(z)-1$.  If $\beta(z)=1$,
then Prover chooses $f(z)$ to be $M$ or $M+1$ based on the parity of $t$, i.e.
so that $f(z)\equiv t~({\rm mod}~2)$.  Thus $|f(z)-M|\leq
1=\beta(z)=\gamma'(z)+1$.  If $\beta(z)=2$, then Prover offers values $M,M+2$ or
$M-1,M+1$, again chosen so that any value $f(z)$ selected by Adversary from the
two values has the right parity, i.e.  $f(z)\equiv t~({\rm mod}~2)$. Thus
$|f(z)-M|\leq 2=\beta(z)=\gamma'(z)+1$, as required.

The rest is exactly as in the proof of Lemma \ref{lem:path2} (with $\gamma$
replaced by $\gamma'$).
\end{proof}

\begin{theorem}\label{thm:npath-odd}
Let $n\geq 5$ be odd. Assume that $\bP_\infty\models \Psi$ and that the vertices
of $\D_\psi$ are properly coloured with colours black and white.  Then the
following are equivalent.

\begin{compactenum}[(I)]
\item $\bP_n\models \Psi$.
\item Prover has a winning strategy in the game $\G(\Psi,\bP_n)$.
\item There are no vertices $u,v$ where $\gamma'(u)\geq \frac{n-1}{2}$ and
$\gamma'(v)\geq \frac{n-1}{2}$ such that $u$ is black and  $v$ is white.
\end{compactenum}
\end{theorem}

\begin{proof}
(I)$\Leftrightarrow$(II) is by Lemma \ref{lem:game}.  For (II)$\Rightarrow$(III)
assume that Prover plays a winning strategy in $\G(\Psi,\bP_n)$, but there are
vertices $u,v$ where $u$ is black, $v$ is white, and $\gamma'(u)\geq
\frac{n-1}2$ and $\gamma'(v)\geq\frac{n-1}2$.  Then Prover's strategy is also
winning on $\bP_\infty$. This implies that the resulting mapping $f$ 
satisfies conditions $(\star)$ and $(\triangle)$.  Moreover, by Lemma
\ref{lem:odd-path1}, Adversary can play so that 
$|f(u)-\frac{n+1}{2}|\geq \gamma'(u)\geq \frac{n-1}{2}$ and that
$|f(v)-\frac{n+1}{2}|\geq\gamma'(v)\geq \frac{n-1}{2}$.  This means that either
$f(u)\geq n$ or $f(u)\leq 1$. Same for $v$.  But since Prover wins, $f$
is a homomorphism to $\bP_n$ and so $f(u)$ and $f(v)$ belong to
$\{1,\ldots,n\}$. Therefore $f(u)$ and $f(v)$ are in $\{1,n\}$.  Recall that $f$
satisfies $(\triangle)$ and that the graph $\D_\psi$ is connected.  Thus there exists
a looping walk from $u$ to $v$ which by ($\triangle$) implies that $f(u)$ and
$f(v)$ have different parity (since $u$ is black while $v$ is white). However,
both $1$ and $n$ are odd numbers and both $f(u)$ and $f(v)$ are from $\{1,n\}$,
which renders this situation impossible.  Therefore Adversary wins showing that
no such vertices $u,v$ exist which proves (III).

Conversely, to prove (III)$\Rightarrow$(II), assume by symmetry that every black
vertex $u$ satisfies $\gamma'(u)\leq \frac{n-3}2$. We first show that this
implies $\gamma'(v)\leq \frac{n-1}2$ for all white vertices $v$. We proceed by
induction on the ordering~$\prec$. Let $v$ be a white vertex. If
$\gamma'(v)=\beta(v)-1$, then $\gamma'(v)=\beta(v)-1\leq 1\leq\frac{n-1}{2}$,
since $n\geq 5$.   So we may assume that there is $u\prec v$ such that
$\gamma'(v)=2(\beta(v)-1)+\gamma'(u)-\delta(u,v)$.
Recall that we assume that $\bP_\infty\models\Psi$. Thus by Theorem
\ref{thm:12csp-path}(IV), we have $\delta(u,v)\geq \beta(v)-1$.
If $u$ is black, we have $\gamma'(u)\leq\frac{n-3}{2}$ and so\smallskip

\centereq{$\gamma'(v)=2(\beta(v)-1)+\gamma'(u)-\delta(u,v)\leq
\frac{n-3}{2}+\beta(v)-1\leq \frac{n-1}{2}$}\smallskip

If $u$ is white, we have by induction (since $u\prec v$) that $\gamma'(u)\leq
\frac{n-1}2$. Moreover,  we observe that $\delta(u,v)$ is even, since both $u$
and $v$ are white. So if $\beta(v)=2$, we never have $\delta(u,v)=\beta(v)-1$
since in that case $\beta(v)-1$ is odd. Thus we can write $\delta(u,v)\geq
2(\beta(v)-1)$ and therefore\smallskip

\centereq{
$\gamma'(v)=2(\beta(v)-1)+\gamma'(u)-\delta(u,v)\leq \frac{n-1}{2}$}\smallskip

This proves that $\gamma'(v)\leq\frac{n-1}{2}$ for all white vertices $v$, as
promised. With this we can describe a winning strategy for Prover.  Let $z$ be
the first vertex in $\prec$. Choose $t=0$ if $z$ is black and $t=1$ if $z$ is
white. 

By Lemma \ref{lem:odd-path2} with the above $t$ and $M=\frac{n+1}{2}$, Prover
can play a strategy in which the resulting mapping $f$ is a homomorphism to
$\bP_\infty$ and where $|f(v)-\frac{n+1}{2}|\leq\gamma'(v)+1$ for every vertex $v$ and
$f(z)\equiv t~({\rm mod}~2)$.  We show that this strategy produces a
homomorphism to $\bP_n$. It suffices to prove that
$f(v)\in\{1,\ldots,n\}$ for every $v$.
If $u$ is black, then $|f(u)-\frac{n+1}{2}|\leq \gamma'(u)+1\leq
\frac{n-3}{2}+1=\frac{n-1}{2}$. Thus $1\leq f(u)\leq n$ as required.

Now consider a white vertex $v$. Since Prover plays a winning strategy on $\bP_\infty$, the
mapping $f$ satisfies $(\triangle)$ and $(\star)$.  We show that this implies
that $f(v)$ is odd.
Recall that $f(z)\equiv t~({\rm mod}~2)$.  Suppose first that $z$ is black. Then
$f(z)$ is even by the choices of $t$. Thus by ($\triangle$), we
deduce that $f(v)$ is odd, since $z$ is black while $v$ is white.  We proceed
similarly if $z$ is white. In this case, $f(z)$ is odd by the choices of $t$,
and so $f(v)$ is also odd by $(\triangle)$, since both $z$ and $v$ are white.
This proves that $f(v)$ is indeed odd for every white vertex $v$, and we recall
 that $\gamma'(v)\leq\frac{n-1}{2}$ as we proved earlier. Therefore, we have
$|f(v)-\frac{n+1}{2}|\leq \gamma'(v)+1\leq \frac{n+1}{2}$. In other words, $0\leq
f(v)\leq n+1$ but since $f(v)$ is odd, we actually have $1\leq f(v)\leq n$, as
required.

This yields (II) and thus completes the proof (III)$\Rightarrow$(II) and the
theorem.
\end{proof}

\subsection{Proof of Theorem \ref{thm:finite-paths}}

We observe that the values $\gamma(v)$ and $\gamma'(v)$ for each vertex $v$ can
be calculated using dynamic programming in polynomial time.  Thus the conditions
of Theorems \ref{thm:npath-even} and \ref{thm:npath-odd} can be tested in
polynomial time and thus Theorem \ref{thm:finite-paths} follows.\hfill\qed

\section{Proof Corollary \ref{cor:finite-forests}}\label{sec:proof-cor3}

We now show how to decide $\{1,2\}$-CSP($H$) when $H$ is a forest.  Let $\Psi$
be a given instance to this problem and let $G=\D_\psi$ be the corresponding
graph.

First, we may assume that $H$ is a tree. This follows easily (with a small
caveat mentioned below) as the connected components of $G$ have to be mapped to
connected components of $H$.
Therefore with $H$ being a tree, we first claim that if $G$ is a yes-instance,
then $G$ is also a yes instance to $\{1,2\}$-CSP($\bP_\infty$).  To conclude
this, it can be shown that the condition (III) of Theorem \ref{thm:12csp-path}
can be generalized to trees by using the distance in $H$ in the condition
($\star$), and using proper colouring of $H$ for ($\triangle$) in an obvious
way.  This implies that no two vertices $u$,$v$ are mapped in $T$ farther away
than $\delta(u,v)$. So a bad walk cannot exist and $G$ is a yes-instance of
$\{1,2\}$-CSP($\bP_\infty$).

A similar argument allows us to generalize Lemmas \ref{lem:path1} and
\ref{lem:path2} to trees; namely Adversary will play {\bf away} from some
vertex, while Prover {\bf towards} some vertex. The absolute values will be
replaced by distances in $H$.  From this we conclude by Theorem
\ref{thm:npath-even} that Adversary can force each $v$ to be assigned a vertex
which is at least $\gamma'(v)$ resp. $\gamma(v)$ away from the center vertex of
$H$, resp. center edge of $H$.  In summary, this proves

\begin{corollary}
Let $H$ be a tree.  Let $P$ be the longest path in $H$.  Then $\Psi$ is a
yes-instance of $\{1,2\}$-CSP($H$) if and only if $\Psi$ is a yes instance of
$\{1,2\}$-CSP$(P)$.
\end{corollary}

This can be phrased more generally for forests in a straightforward manner. The
only caveat is that if two components contain a path of the same length, we can
make the first vertex in the instance an $\exists^{\geq 1}$ vertex without affecting
the satisfiability, because if it is $\exists^{\geq 2}$, we let Adversary choose 
which midpoint of the two longest paths to use (which is symmetric).\hfill\qed
\bigskip

Finally, we are ready to prove a dichotomy for $\{1,2\}$-CSP($H$) where $H$ is a graph.

\section{Proof of Corollary \ref{cor:12csp-dichotomy}}\label{sec:proof-cor4}

If $H$ is not bipartite, then $\{1\}$-CSP($H$) is NP-hard by
\cite{HellNesetril}; thus $\{1,2\}$-CSP($H$) is also NP-hard.  So we may assume
that $H$ is a bipartite graph. If $H$ contains a $C_4$, then $\{1,2\}$-CSP($H$)
is in $L$ as shown in \cite{csr2012}. If $H$ contains a larger cycle, then the
problem is Pspace-complete as we show later in Theorem \ref{thm:bip-dichotomy}.
Thus we may assume that $H$ contains no cycles and hence it is a forest. In
this case $\{1,2\}$-CSP($H$) is polynomial-time solvable as shown in Corollary
\ref{cor:finite-forests}.  That completes the proof.\hfill\qed

\section{Pspace dichotomy for bipartite graphs -- Proof of Theorem
\ref{thm:bip-dichotomy}}\label{sec:pspace-dich}

The following is a slightly simpler version of a subcase
of Proposition 5~from~\cite{csr2012}.

\begin{proposition}\label{prop:pspace-1}
If $j \geq 3$, then $\{1,2\}$-CSP$(\bC_{2j})$ is Pspace-complete.
\end{proposition}

\begin{proof}
We reduce from QCSP$({\bK}_j)$, whose Pspace-completeness follows from
\cite{BBCJK}. Our reduction borrows a lot from the reduction from
CSP$({\bK}_j)$ to the retraction problem Ret$({\bC}_{2j})$ used to
prove the NP-hardness of the latter in \cite{FederHellHuang99}. 
For an input $\Psi:=Q_1 x_1 Q_2 x_2 \ldots Q_m x_m \ \psi(x_1,x_2,\ldots,x_m)$
for QCSP$({\bK}_j)$ we build an input $\Theta$ for
$\{1,2\}$-CSP$({\bC}_{2j})$ as follows. We begin by considering the graph
$\mathcal{D}_\psi$, from which we first build a graph
${G}':=\mathcal{D}_\psi \uplus {\bC}_{2j}$, with the latter cycle
on new vertices $\{w_1,\ldots,w_{2j}\}$. Now we build ${G}''$ from
${G}'$ by replacing every edge $(x,y) \in \mathcal{D}_\psi$ with a
gadget which involves $3j$ new copies of ${\bC}_{2j}$ connected in a
prismic path (cartesian product with ${\bP}_{3j}$) to the fixed copy of
${\bC}_{2j}$ in ${G}'$ -- induced by $\{w_1,\ldots,w_{2j}\}$. The
vertex $x$ sits at the end of a pendant path of length $j-1$ which joins the
final, leftmost copy of ${\bC}_{2j}$ diametrically opposite $y$; and the
shortest path from $x$ to $y$ is of length $j-1$ (if $j$ odd) or $j$ (if $j$
even). Finally, for universal variables $v$ of $\Psi$ we add a new path
$v_1,\ldots,v_j$ culminating in $v$ of length $j$. These gadgets are drawn in
\cite{csr2012}.

The quantifier-free part $\theta$ of $\Theta$ will be $\phi_{{G}''}$ we
now explain how to add the quantification. The variables $\{w_1,\ldots,w_{2j}\}$
will be quantified outermost and last. Aside from this, and moving inwards
through the quantifier order of $\Psi$, when we encounter an existential
variable $v$, we apply existential quantification to it in $\Theta$. On the
other hand, when we encounter a universal variable $v$, we apply existential (or
$\exists^{\geq 2}$) quantification to the first element $v_1$ of its associated
path, followed by $\exists^{\geq 2}$ quantification to the remaining
$v_2,\ldots,v_j$ of this path and finally existential quantification to the $v$
that this path culminates in. All the remaining parts of the edge gadgets can be
existentially quantified innermost.

Suppose that the variables $w_1,\ldots,w_{2j}$ are evaluated truly, \mbox{i.e.}
up to some automorphism of ${\bC}_{2j}$. It is simple to see that the
gadgets enforce that variables $x$ and $y$ of $\Psi$ are evaluated at distinct
points of ${\bC}_{2j}$ with the same parity. Furthermore, the universally
quantified variables are forced, under suitable evaluation of the $\exists^{\geq
2}$ variables within their gadgets, to all positions of ${\bC}_{2j}$ of
the respective parity.

Finally we need to enforce on the cycle ${\bC}_{2j}$ with the variables
$w_1,\ldots,w_{2j}$, the outermost quantification \[\exists^{\geq 2} w_1, w_2,
w_3, \ldots, w_{j+1} \exists w_{j+2},\ldots,w_{2j-1}. \] The point is that there
is some evaluation of $w_1, \ldots, w_{j+1}$ that enforces that $w_1,
\ldots,w_{2j}$ is evaluated automorphically to ${\bC}_{2j}$. Clearly, the
outermost quantification on $w_1$ and $w_2$ could be pure existential. The other
possibilities (the \emph{degenerate} cases) involve a mapping to a homomorphic
image of  ${\bC}_{2j}$ which  we can extend to homomorphism through
careful assignment of the parts of the chains of the edge gadgets in which no
$\exists^{\geq 2}$ quantification subsists. The degenerate cases are the reason
our chain is of length $3j$ rather than the much shorter $j-1$ that appears in
\cite{FederHellHuang99}; for when  $w_1, \ldots, w_{2j}$ are mapped
degenerately, the universal variables can still appear anywhere on the cycle.
\end{proof}

\begin{corollary}
If $H$ is a bipartite graph whose smallest cycle in $\bC_{2j}$
for $j \geq 3$, then $\{1,2\}$-CSP$(H)$ is Pspace-complete.
\end{corollary}

\begin{proof}
There are two difficulties arising from simply using the proof of the previous
proposition as it is. Firstly, let us imagine that the variables  $w_1, \ldots,
w_{2j}$ are indeed mapped to a fixed copy of $\bC_{2j}$ in
$H$. We need to ensure that variables $x$, $y$ derived from the
original instance $\Psi$ are mapped to the cycle also. For $y$ variables in our
gadget one can check this must be true -- the successive cycles in the edge
gadget may never deviate from the fixed  $\bC_{2j}$, since $H$
contains no smaller cycle than  $\bC_{2j}$ -- but for $x$ variables off
on the pendant this might not be true. There are various fixes for this. If the
instance of $\Psi$ were symmetric, \mbox{i.e.} contained an atom $E(x,y)$ iff it
also contained $E(y,x)$ then this property would automatically hold (and it is
easy to see from \cite{BBCJK} that  QCSP$(\bK_j)$ remains
Pspace-complete on symmetric instances). An alternative is to add a system of
$3j$ cycles to tether variables $x$ also to the fixed copy of
$\bC_{2j}$.

The second difficulty is in isolating a copy of the cycle $\bC_{2j}$
with the variables $w_1,\ldots,w_{2j}$, but since $H$ does not contain
a cycle smaller than $\bC_{2j}$ a simple argument shows that one such
cycle must be identified.\end{proof}\vspace{-2ex}

\section{Cases in NP: Dominating vertices}\label{sec:domin}
 
We consider here the class of undirected graphs with a single dominating vertex
$w$ which is also a self-loop.

\begin{proposition}\label{prop:dom-bipartite}
If $H$ has a reflexive dominating vertex $w$ and $H
\setminus \{w\}$ contains a loop or is irreflexive bipartite, then
$\{1,2\}$-CSP$(H)$ is in P.
\end{proposition}

\begin{proof}
If ${H} \setminus \{w\}$ contains a loop then ${H}$ contains
adjacent looped vertices and $\{1,2\}$-CSP$({H})$ is trivial (all
instances are yes-instances). Assume ${H} \setminus \{w\}$ is
irreflexive bipartite and consider an input $\Psi$ for
$\{1,2\}$-CSP$({H})$. All variables quantified by $\exists$ can be
evaluated as $w$ and can be safely removed while preserving satisfaction. So,
let $\Psi'$ be the subinstance of $\Psi$ induced by the variables quantified by
$\exists^{\geq 2}$ and let $\psi'$ be the associated quantifier-free part. If
$\mathcal{D}_{\psi'}$ is bipartite, the instance is a yes-instance, otherwise it
is a no-instance.~\end{proof}

\begin{proposition}\label{prop:dom-NP-complete}
If $H$ has a reflexive dominating vertex $w$ and $H
\setminus \{w\}$ is irreflexive non-bipartite, then $\{1,2\}$-CSP$(H)$
is NP-complete.
\end{proposition}

\begin{proof}
For membership of NP we note the following. Let $U$ be a unary predicate
defining the set ${H} \setminus \{w\}$. From an input $\Psi$ for
$\{1,2\}$-CSP$({H})$ we will build an instance $\Psi'$ for
CSP$({H};U)$ so that ${H} \models \Psi$ iff $({H};U)
\models \Psi'$. The latter is clearly in NP, so the result follows. To build
$\Psi'$ we take $\Psi$ and add $U(v)$ to the quantifier-free part for all
$\exists^{\geq 2}$ quantified variables $v$, before converting in the
quantification $\exists^{\geq 2} v$ to $\exists v$.

For NP-hardness we reduce from CSP(${H} \setminus \{w\}$) which is
NP-hard by \cite{HellNesetril}. For an input $\Psi$ to this, we build an input
$\Psi'$ for $\{1,2\}$-CSP$({H})$ by converting each $\exists$ to
$\exists^{\geq 2}$. It is easy to see that $({H} \setminus \{w\})
\models \Psi$ iff ${H} \models \Psi'$ and the result follows.
\end{proof}

\begin{corollary}
If $H$ has a reflexive dominating vertex, then
$\{1,2\}$-CSP$(H)$ is either in P or is NP-complete.
\end{corollary}

\section{Small graphs}\label{sec:small}

It follows from Proposition~\ref{prop:dom-NP-complete} that there is a partially
reflexive graph on four vertices, $\bK_4$ with a single reflexive
vertex, so that the corresponding  $\{1,2\}$-CSP is NP-complete. We can argue
this phenomenon is not visible on smaller graphs.

\begin{proposition}\label{prop:small}
Let $H$ be a (partially reflexive) graph on at most three vertices,
then either $\{1,2\}$-CSP$(H)$ is in P or it is Pspace-complete.
\end{proposition}

\begin{proof}
The Pspace-complete cases are ${\bK}_3$ (see \cite{csr2012})
and $\mathcal{P}_{101}$, which is the path of length two whose internal vertex
is loopless while the end vertices are looped. It is known
QCSP$(\mathcal{P}_{101})$ is Pspace-complete \cite{QCSPforests}. One can reduce
this problem to $\{1,2\}$-CSP$(\mathcal{P}_{101})$ by substituting $\forall x$
in the former by $\exists^{\geq 2} x,x' \ E(x,x')$ in the latter (where $x'$ is
a newly introduced variable).

It is clear that $\{1,2\}$-CSP$({H})$ is trivial if ${H}$
contains a reflexive clique of size $2$, ${\bK}_2^*$. If ${H}$ is
irreflexive bipartite, \mbox{i.e.} is a forest, then
$\{1,2\}$-CSP$({H})$ is in P according to \cite{csr2012}.
When ${H}$ contains just isolated loops and non-loops then it is easy to
give a tailored algorithm. If ${H}$ contains at least two loops then:
any input with a subinstance $Q x \exists^{\geq 2} x' \ E(x,x')$ ($Q$ any
quantifier, $x\neq x'$) is false; and all other inputs are true. If
${H}$ contains only one loop:  any input with a subinstance $Q x
\exists^{\geq 2} x' \ E(x,x')$ ($Q$ any quantifier, possibly $x=x'$) is false;
and all other inputs are true. If ${H}$ contains just isolated loops
then it is bipartite. We henceforth assume these cases solved.

We are left with one remaining two-vertex graph, $\mathcal{P}_{10}$. For this
problem, any input with a subinstance $\exists^{\geq 2} x, x' \ E(x,x')$ ($Q$
any quantifier, possibly $x=x'$) is false; and all other inputs are true.

We continue with graphs of exactly three vertices. Among the remaining
possibilities where ${H}$ has exactly two loops is only
$\mathcal{P}_{10} \uplus {\bK}_1^*$. For this problem, any input with a
subinstance $\exists^{\geq 2} x, x' \ E(x,x')$ ($Q$ any quantifier, $x \neq x'$)
is false; and all other inputs are true.

We now address the case in which there is precisely one loop. If it dominates,
then we have tractability by Proposition~\ref{prop:dom-bipartite}. If it is
isolated, then the remaining case is ${\bK}_1^* \uplus {\bK}_2$.
For this, any input $\Psi$ with subinstance $\exists^{\geq 2} x \ E(x,x)$ or an
$\exists^{\geq 2} v$ attached to a sequence (connected component of $v$ in
$\mathcal{D}_\psi$) that is non-bipartite is false; and all other inputs are
true. The remaining possibilities are $\mathcal{P}_{100}$ and $\mathcal{P}_{10}
\uplus {\bK}_1$. For the latter, we have the same $\{1,2\}$-CSP as for
$\mathcal{P}_{10}$, which has already been resolved.
$\{1,2\}$-CSP$(\mathcal{P}_{100})$ requires some subtlety and appears as its own
result in Proposition~\ref{prop:P100}.

Finally, we come to the irreflexive cases, and realise these are either
bipartite or ${\bK}_3$ and are hence resolved.
\end{proof}

We denote $\mathcal{P}_{100}$ the path on three vertices $0,1,2$ with loop at $0$.

\begin{proposition}\label{prop:P100}
 $\{1,2\}$-CSP$(\mathcal{P}_{100})$ is in P.
\end{proposition}

\begin{proof}
The following four types of subinstance in $\Psi$ result in it being false
(always consider the symmetric closure).
\begin{itemize}
\item[$(i.)$] $\exists^{\geq 2} x_1,x_2,x_3 \ E(x_1,x_2) \wedge E(x_2,x_3)$.
\item[$(ii.)$] $Q x_1 \exists^{\geq 2} x_2,x_3 \ E(x_1,x_3) \wedge E(x_2,x_3)$ ($Q$ any quantifier).
\item[$(iii.)$] $\exists^{\geq 2} x_1,x_2,x_3 \ E(x_2,x_3) \wedge E(x_1,y) \wedge E(x_3,y)$; where $y$ is quantified anywhere existentially.
\item[$(iv.)$] $\exists^{\geq 2} x_1,x_2,x_3,x_4 \ E(x_1,x_2) \wedge E(x_3,x_4) \wedge E(x_2,y_1) \wedge E(x_4,y_2) \wedge E(y_1,y_2)$; where $y_1,y_2$ are quantified anywhere existentially.
\end{itemize}
We claim that all other inputs are yes-instances. We give the following strategy
for Prover. Consider ${\cal P}_{100}$ to be $\{0,1,2\}$ and take the canonical sequence
$0,1,2$ to be the path of $\mathcal{P}_{100}$ from the loop $0$. For
$\exists^{\geq 2}$ variables Prover offers $\{0,1\}$, unless constrained by
adjacency of a variable already played to $1$, in which case she offers
$\{0,2\}$. For $\exists$ variables Prover offers $\{0\}$, unless constrained by
adjacency of a variable already played to $2$, in which case she offers $\{1\}$.
We argue this strategy must be winning. This is tantamount to saying that Prover
is never offered (A.) an $\exists^{\geq 2}$ variable that is simultaneously
adjacent to $0$ and $1$; (B.) an $\exists^{\geq 2}$ variable that is adjacent to
$2$; and (C.) an $\exists$ variable that is simultaneously adjacent to both $2$
and $1$. (B) follows from Rule $(i)$ and (A) follows from Rule $(ii)$. (C)
follows from $(iii)$ and $(iv)$. 
\end{proof}

\newcommand{\nosubset}{~\makebox[0.1cm][l]{\ensuremath{\subseteq}}\!/~}

\section{Final remarks}\label{sec:final}
In this paper we have settled the major questions left open in \cite{csr2012}
and it might reasonably be said we have now concluded our preliminary
investigations into constraint satisfaction with counting quantifiers. Of course
there is still a wide vista of work remaining, not the least of which is to
improve our P/ NP-hard dichotomy for $\{1,2\}$-CSP on undirected graphs to a P/
NP-complete/ Pspace-complete trichotomy (if indeed the latter exists). The
absence of a similar trichotomy for QCSP, together with our reliance on
\cite{HellNesetril}, suggests this could be a challenging task.  Some more
approachable questions include lower bounds for $\{2\}$-CSP($\bK_4$) and
$\{1,2\}$-CSP($\bP_\infty$). For example, intuition suggests these might be
NL-hard (even P-hard for the former). Another question would be to study
$X$-CSP($\bP_\infty$), for $\{1,2\} \nosubset X \subset \mathbb{N}$. 

Since we initiated our work on constraint satisfaction with counting
quantifiers, a possible algebraic approach has been published in
\cite{BulatovHedayatyISMVL2012,BulatovHedayatyArxiv2012}. It is clear reading
our expositions that the combinatorics associated with our counting quantifiers
is complex, and unfortunately the same seems to be the case on the algebraic
side (where the relevant "expanding" polymorphisms have not previously been
studied in their own right). At present, no simple algebraic method,
generalizing results from \cite{BBCJK}, is known for counting quantifiers with
majority operations. This would be significant as it might help simplify our
tractability result of Theorem~\ref{thm:finite-paths}. So far, only the Mal'tsev
case shows promise in this direction. 

\bibliographystyle{acm}
\bibliography{local}

\end{document}